\documentclass[a4paper,pdfa,UKenglish,cleveref, autoref, thm-restate]{lipics-v2021}

\pdfoutput=1 %
\hideLIPIcs  %

\usepackage{todonotes}
\usepackage{xspace}
\usepackage{macros}
\usepackage[linesnumbered,ruled,vlined]{algorithm2e}
\usepackage{comment}
\title{FPT Approximation of Generalised Hypertree Width for Bounded Intersection Hypergraphs} 

\titlerunning{FPT Approximation of Generalised Hypertree Width for Bounded Intersection} 

\author{Matthias Lanzinger}{TU Wien \& University of Oxford}{matthias.lanzinger@tuwien.ac.at}{https://orcid.org/0000-0002-7601-3727}{Matthias Lanzinger acknowledges support by the Royal Society  “RAISON DATA” project (Reference No. RP\textbackslash{}R1\textbackslash{}201074) and the Vienna Science and Technology Fund (WWTF) [10.47379/ICT2201, 10.47379/VRG18013, 10.47379/NXT22018].}

\author{Igor Razgon}{Birkbeck, University of London}{igor@dcs.bbk.ac.uk}{https://orcid.org/0000-0002-7060-5780}{}

\authorrunning{M. Lanzinger \& I. Razgon}

\Copyright{CC-BY}



\nolinenumbers %

\EventEditors{Olaf Beyersdorff, Mamadou Moustapha Kant\'{e}, Orna Kupferman, and Daniel Lokshtanov}
\EventNoEds{4}
\EventLongTitle{41st International Symposium on Theoretical Aspects of Computer Science (STACS 2024)}
\EventShortTitle{STACS 2024}
\EventAcronym{STACS}
\EventYear{2024}
\EventDate{March 12--14, 2024}
\EventLocation{Clermont-Ferrand, France}
\EventLogo{}
\SeriesVolume{289}
\ArticleNo{42}

\newcommand{\ghw}{\ensuremath{\mathit{ghw}}\xspace}
\newcommand{\hw}{\ensuremath{\mathit{hw}}\xspace}
\newcommand{\tw}{\ensuremath{\mathit{tw}}\xspace}
\newcommand{\np}{\textsf{NP}\xspace}
\newcommand{\fpt}{\textsf{FPT}\xspace}
\newcommand{\akd}{\ensuremath{\alpha^\times_{k,d}}\xspace}

\newcommand{\bep}{\ensuremath{\mathsf{BE}_p}\xspace}
\newcommand{\B}{\ensuremath{\mathbf{B}}\xspace}
\newcommand{\Reject}{\ensuremath{\mathbf{Reject}}\xspace}
\newcommand{\sep}{\mathit{sep}\xspace}
\newcommand{\nop}[1]{}
\newcommand{\gcaalg}{\ensuremath{\mathit{GapCoverApprox}}\xspace}

\begin{document}

\maketitle

\begin{abstract}

Generalised hypertree width (\ghw) is a hypergraph parameter that is central to the tractability of many prominent problems with natural hypergraph structure. Computing \ghw of a hypergraph is notoriously hard. The decision version of the problem, checking whether $\ghw(H) \leq k$, is \textsc{paraNP}-hard when parameterised by $k$. Furthermore, approximation of $\ghw$ is at least as hard as approximation of \textsc{Set-Cover}, which is known to not admit any \fpt approximation algorithms.

Research in the computation of \ghw so far has focused on identifying structural restrictions to hypergraphs -- such as bounds on the size of edge intersections --  that permit \textsc{XP} algorithms for \ghw. Yet, even under these restrictions that problem has so far evaded any kind of \fpt algorithm.
In this paper we make the first step towards \fpt algorithms for \ghw by showing that the parameter can be approximated in \fpt time for graphs of bounded edge intersection size.
In concrete terms we show that there exists an \fpt algorithm, parameterised by $k$ and $d$, that
for input hypergraph $H$ with maximal cardinality of edge intersections $d$ and integer $k$ either outputs a tree decomposition with $\ghw(H) \leq 4k(k+d+1)(2k-1)$, or rejects, in which case it is guaranteed that $\ghw(H) > k$. Thus, in the special case of hypergraphs of bounded edge intersection, we obtain an \fpt $O(k^3)$-approximation algorithm for $\ghw$.
 \end{abstract}

\newpage

\section{Introduction}
\label{sec:intro}

A tree decomposition of a hypergraph $H$ is a pair $(T, \B)$ where $T$ is a tree and ${\bf B}: V(T) \to 2^{V(H)}$ assigns a \emph{bag} to each node, that satisfies certain properties. The treewidth of a decomposition is $\max_{u\in V(T)} |\B(u)|-1$ and the treewidth of $H$ is the least treewidth taken over all decompositions. In many graph algorithms, treewidth is the key parameter that determines the complexity of the problem.
However, for many problems whose underlying structure is naturally expressed in terms of hypergraphs the situation is different. The treewidth of a hypergraph is always at least as large as its \emph{rank} ($-1$), i.e., the maximal size of an edge. Yet, many standard hypergraph problems can be tractable even with unbounded rank. To counteract this problem, \emph{generalised hypertree width} (\ghw) often takes the place of treewidth in these cases. The definition of \ghw is also based on tree decompositions, with the only difference being that the \ghw of a decomposition is $\max_{u\in V(T)} \rho(\B(u))$, where $\rho(U)$ is the least number of edges required to cover $U \subseteq V(H)$.
Parallel to treewidth in graphs, low \ghw is a key criterion for tractability in the hypergraph setting.
Prominent examples include the evaluation of conjunctive queries, database factorisation~\cite{DBLP:journals/sigmod/OlteanuS16}, winner determination in combinatorial auctions~\cite{DBLP:journals/jacm/GottlobG13}, and determining Nash Equilibria in strategic games~\cite{DBLP:journals/jair/GottlobGS05}. 

Computation of \ghw is computationally challenging. The problem is known to be \textsf{paraNP}-hard~\cite{DBLP:journals/jacm/GottlobMS09,DBLP:conf/pods/FischlGP18} and \textsc{W[2]}-hard~\cite{DBLP:conf/wg/GottlobGMSS05} in the parameterised setting\footnote{When discussing the parameterised complexity of deciding whether a width parameter is at most $k$, we always refer to the parameterisation by $k$ if not specified otherwise.} (see the discussion of related work below for details). This is shown by reduction from \textsc{Set-Cover}, which together with recent breakthrough results on the approximability of \textsc{Set-Cover}~\cite{DBLP:journals/jacm/SLM19}, also implies that there can be no \fpt approximation algorithms for \ghw under standard assumptions.
In this paper, we introduce the first \fpt algorithm for approximation of \ghw in unbounded rank hypergraphs. Notably, this comes over 20 years after the parameter was first introduced ~\cite{DBLP:journals/jcss/GottlobLS03}. As there can be no such algorithm in the general case, we instead consider the restriction of the cardinality of the intersections of any two edges. Formally, a $(2,d)$-hypergraph is a hypergraph where the intersection of any pair of edges has cardinality at most $d$.
We will study the following problem $f$-\textsc{ApproxGHW}.
\begin{problem}{$f$-ApproxGHW}
  Input & $(2,d)$-hypergraph $H$, positive integer $k$ \\
  Parameters & $k$ and $d$\\
  Output & A tree decomposition of $H$ with $\ghw$ at most $f(k, d)$,
  \\ & or \textbf{Reject}, in which case $\ghw(H) > k$.
\end{problem}
Let $\alpha(k, d)$ refer to the term $k(3k+d+1)(2k-1)$. Our main result is the following.
\begin{restatable}{theorem}{RESTATEmain}
 \label{thm:approxghw}
 \textsc{$4\alpha(k, d)$-ApproxGHW} is fixed-parameter tractable,
\end{restatable}

Our algorithm follows classic ideas for \fpt algorithms for treewidth but requires significant new developments. Most critically, we propose an \fpt algorithm for computing approximate $(A, B)$-separators in $(2,d)$-hypergraph, i.e., set of vertices $S$ such that sets of vertices $A$ and $B$ are not connected in $H$ without $S$. 

\textbf{Related Work.}
Despite the close relationship between \ghw and treewidth, there is a stark difference in the complexity of recognising the respective widths. While it is famously possible to decide $\tw(\cdot) \leq k$ in fixed-parameter linear time~\cite{DBLP:journals/siamcomp/Bodlaender96}, deciding $\ghw(\cdot) \leq k$ is significantly harder. Intuitively, this is because techniques for efficiently deciding treewidth fundamentally rely on bounding the number of vertices in the bag, yet even $\alpha$-acyclic hypergraphs (those with $\ghw$ 1) can require decompositions with arbitrarily large bags to achieve minimal \ghw.
In concrete terms, deciding $\ghw(\cdot)\leq k$ has been shown to be \np-hard for all fixed $k > 1$ (or \textsf{paraNP}-hard in terms of parameterised complexity)~\cite{DBLP:journals/jacm/GottlobMS09,DBLP:conf/pods/FischlGP18}. Additionally, deciding $\ghw(\cdot) \leq k$ is known to be \textsf{W[2]}-hard by a reduction from \textsc{Set Cover}~\cite{DBLP:conf/wg/GottlobGMSS05}. 

In response, significant effort has been invested in identifying conditions under which deciding $\ghw(\cdot) \leq k$ is tractable~\cite{DBLP:journals/jacm/GottlobMS09,DBLP:journals/jacm/GottlobLPR21} for fixed $k$ (i.e., the problem is in \textsf{XP}). Of particular note here is the observation that the problem is in \textsf{XP} if we restrict the problem to so-called $(c, d)$-hypergraphs, i.e., hypergraphs where any intersection of at least $c$ edges has cardinality at most $d$. Notably, this coincides with the most general condition known to allow kernelization for \textsc{Set-Cover}~\cite{DBLP:conf/esa/PhilipRS09}.

As mentioned above, negative results for \fpt approximation of \textsc{Set-Cover} apply also to approximating \ghw. 
Nonetheless, some approximation results are known when looking beyond \fpt. Importantly, \ghw is 3-approximable in \textsf{XP} via the notion of (not generalised) \emph{hypertree width} (\hw)~\cite{DBLP:journals/jcss/GottlobLS03,DBLP:journals/ejc/AdlerGG07}. Despite this improvement in the case of fixed $k$, deciding hypertree width is also \textsf{W[2]}-hard~\cite{DBLP:conf/wg/GottlobGMSS05} by the same reduction from \textsc{Set-Cover} as \ghw. \emph{Fractional hypertree width} $\mathit{fhw}$ generalises \ghw in the sense that the width is determined by the \emph{fractional cover number} of the bags~\cite{DBLP:journals/talg/GroheM14}. This width notion is strictly more general than \ghw and allows for cubic approximation on $\mathsf{XP}$~\cite{DBLP:journals/talg/Marx10}.
Recently, Razgon~\cite{DBLP:journals/corr/abs-2212-13423} proposed an \fpt algorithm for constant factor approximation of \ghw for hypergraphs of bounded rank. While we consider this to be an important conceptual step towards our result, it should be noted that for any hypergraph $H$ it holds that $\ghw(H) \leq \tw(H) \cdot \mathit{rank}(H)$, i.e., assuming bounded rank tightly couples the problem to deciding treewidth.

\textbf{Structure of the paper.}
Our main result combines on multiple novel combinatorial observations and \fpt algorithms and the main body of this paper is therefore focused on presenting the main ideas and how these parts interact. Full proof details are provided in the appendix.
After basic technical preliminaries in \Cref{sec:prelim}, we give a high-level overview of the individual parts that lead to the main result in \Cref{sec:overview}. After the initial overview, \Cref{sec:pseudo} presents the key algorithms formally, together with sketches of their correctness and complexity. Similarly, the main combinatorial ideas are discussed in \Cref{sec:combstat}. We discuss potential avenues for future research in \Cref{sec:conclusion}. 
 
\section{Preliminaries}
\label{sec:prelim}

We will frequently write $[n]$ for the set $\{1,2,\dots,n\}$.
We say that (possibly empty) pairwise disjoint sets $X_1, X_2, \dots, X_n$ are a \emph{weak partition} of set $X$ if their union equals $X$. We assume familiarity with standard concepts of parameterised algorithms and refer to~\cite{DBLP:books/sp/CyganFKLMPPS15} for details. We use $\mathit{poly}(\cdot)$ to represent some polynomial function in the representation size of the argument. 

A \emph{hypergraph} $H$ is a pair of sets $(V(H), E(H))$ where we call $V(H)$ the \emph{vertices} of $H$ and $E(H) \subseteq 2^{V(H)}$ the \emph{(hyper)edges} of $H$. We assume throughout that $H$ has no isolated vertices, i.e., vertices that are not in any edge. For $v \in V(H)$, define $I(v) := \{e \in E(H) \mid v \in e\}$, i.e., the set of edges incident to $v$.
We say that $H$ is a $(c,d)$-hypergraph if for any set $\{e_1,\dots,e_c\} \subseteq E(H)$ of $c$ edges, it holds that $|\bigcap_i^c e_i| \leq d$. 
We will refer to the set of (maximal) connected components of  a hypergraph $H$  as $\mathit{CComp}(H)$.
The \emph{induced subhypergraph} of $H$ induced by $U \subseteq V(H)$ is the hypergraph $H'$ with $V(H')=U$ and $E(H')=\{e \cap U \mid e \in E(H)\} \setminus  \emptyset$. We use the notation $H[U]$ to mean the induced subhypergraph of $H$ induced by $U$. For tree $T$ and $v \in V(T)$ we sometimes write $T - v$ to mean the subgraph of $T$ obtained by deleting $v$ and its incident edges.
Let $A, B \subseteq V(H)$. An $(A,B)$-separator is a set $S \subseteq V(H)$ such that there is no path from an $a \in A \setminus S$ to a $b \in B \setminus S$ in $H[V(H)\setminus S]$.
For $e \in E(H)$, $U_1,\dots,U_n\subseteq V(H)$ we say that $e$ \emph{touches} $U_1,\dots,U_n$ if $e \cap U_i \neq \emptyset$ for all $i\in[n]$. 

An \emph{edge cover} $\mu$ for $U \subseteq V(H)$ is a subset of $E(H)$ such that $U \subseteq \bigcup \mu$. We sometimes refer to the cardinality of an edge cover as its \emph{weight}. The \emph{edge cover number} $\rho(U)$ for set $U \subseteq V(H)$ is the minimal weight over all edge covers for $U$. We sometimes say that $\mu$ is an edge cover of $H$ to mean an edge cover of $V(H)$. Similarly, we use $\rho(H)$ instead of $\rho(V(H))$. A set of edges $E' \subseteq E(H)$ is \emph{$\rho$-stable} if $E'$ is a minimal weight cover for $\bigcup E'$.

A set of sets $S_1,\dots$ can naturally be interpreted as a hypergraph, by considering each set $S_i$ as an edge. In that light, it is clear that deciding $\rho(H) \leq k$ is precisely the same as deciding whether a set system admits a set cover of size $k$. 
\begin{proposition}[\cite{DBLP:conf/stacs/FabianskiPST19}]
\label{fptsetcov}
    There is an \fpt algorithm parameterised by $k+c+d$ that decides for a given $(c,d)$-hypergraphs $H$ and $k \geq 1$ whether $\rho(H) \leq k$.
\end{proposition}

A \emph{tree decomposition} (TD) of hypergraph $H$ is a pair $(T,\B)$ where $T$ is a tree and $\B: V(T) \to 2^{V(H)}$ labels each node of $T$ with its so-called \emph{bag}, such that the following hold:
\begin{romanenumerate}
    \item for each $e \in E(H)$, there is a $u \in V(T)$ such that $e \subseteq \B(u)$, and
    \item for each $v \in V(H)$, the set $\{u \in V(U) \mid v \in \B(u)\}$ induces a non-empty subtree of $T$.
\end{romanenumerate}
We refer to the first property as the \emph{containment property} and to the second as the \emph{connectedness condition}.
For a set $U \subseteq V(T)$ we use $\B(U)$ as a shorthand for $\bigcup_{u \in U} \B(u)$, i.e., all the vertices that occur in bags of nodes in $U$. Similarly, for subtree $T'$ of $T$, we sometimes use $\B(T')$ instead of $\B(V(T'))$.
The \emph{generalised hypertree width} (\ghw) of a tree decomposition is $\max_{u \in V(T)} \rho(\B(u))$, and the generalised hypertree width of $H$ (we write $\ghw(H)$) is the minimal \ghw over all tree decompositions for $H$. It will be important to remember at various points that $\ghw$ is monotone under taking induced subhypergraphs, i.e., $\ghw(H[U]) \leq \ghw(H)$ for all $U\subseteq V(H)$.
 
\section{High-level Overview of Algorithm}
\label{sec:overview}
The algorithm for \textsc{$4\alpha(k,d)$-ApproxGHW}
is based on a standard approach for treewidth computation \cite{DBLP:journals/jal/RobertsonS86}.
However, this approach is in fact applied for a tree decomposition \emph{compression}
rather than approximation from scratch. 
In other words, the actual problem being solved is the following one. 

\begin{problem}{Compress}
  Input & A $(2,d)$-hypergraph $H$, integer $k$, \\
  & a TD $(T, \B)$ of $H$ with \ghw $4\alpha(k,d)+1$ , \\
  & $W \subseteq V(H)$ with $\rho(W) \leq 3\alpha(k,d)$ \\
  Parameters & $k$ and $d$\\
  Output & A tree decomposition $(T^*, \B^*)$ of $H$ with $\ghw$ at most $4\alpha(k,d)$ such that $W \subseteq \B^*(u)$ for some $u \in V(T^*)$, \\
  & or \textbf{Reject}, in which case $\ghw(H) > k$.
\end{problem}

\begin{restatable}{theorem}{RESTATEfptcompress}
  \label{thm:compress}
  \textsc{Compress} is fixed-parameter tractable.
\end{restatable}

One natural question is how the \textsc{Compress} being \fpt
implies \textsc{$4\alpha(k,d)$-ApproxGHW}. This is done through the use
of \emph{iterative compression}, a well known methodology for the design of
\fpt algorithms. The resulting algorithm for \Cref{thm:approxghw} is presented in \Cref{alg:main}. In particular, we let $V(H)=\{v_1, \dots, v_n\}$,
set $V_i=\{v_1, \dots,v_i\}$ and $H_i=H[V_i]$ and solve the \textsc{$4\alpha(k,d)$-ApproxGHW}
for graphs $H_1, \dots, H_n$. 
If some intermediate $H_i$ is rejected, the whole $H$
can be rejected. Otherwise, the application to $H_i$ results in a tree decomposition
$(T_i,{\bf B}_i)$ of \ghw at most $4\alpha(k,d)$. Add $v_{i+1}$ to each bag of 
$(T_i,{\bf B}_i)$. If the \ghw of the resulting decomposition is still at most $4\alpha(k,d)$
simply move on to the next iteration. Otherwise, apply the algorithm for \textsc{Compress}
rejecting if the algorithm rejects and moving to the next iteration if a compressed
tree decomposition is returned.

\begin{algorithm}[t]
    \DontPrintSemicolon
    \SetKwInOut{Input}{Input}\SetKwInOut{Output}{Output}
    \Input{$(2,d)$-hypergraph $H$, positive integer $k$}
    $\{v_1,\dots,v_n\} \gets V(H)$\;
    Let $T_1$ be the tree with a single node $r$\;
    Let ${\bf B}_1$ be the function $r \mapsto \{v_1\}$\;
    \For{$1 < i \leq n$}{
      $V_i \gets \{v_1, \dots, v_i\}$\;
      $H_i \gets H[V_i]$ \;
      $T_i \gets T_{i-1}$\;
      ${\bf B}_i \gets \{t \mapsto {\bf B}_{i-1}(t) \cup \{v_i\} \mid t \in T_i\}$\;
      \If{$\ghw((T_i,{\bf B}_i)) > 4\alpha(k,d)$}{
        $X \gets Compress(H_i, k, (T_i, {\bf B}_i), \emptyset)$\;
        \If{$X$ is \textbf{Reject}}{
          \KwRet{\textbf{Reject}}\;
        }
        $(T_i,B_i) \gets X$\;
      }
    }
    \KwRet{$(T_n,{\bf B}_n)$}\;
    \caption{An \fpt algorithm for $4\alpha(k,d)$-\textsc{ApproxGHW}.}
    \label{alg:main}
  \end{algorithm}

Let us turn our attention to the algorithm for \textsc{Compress}.
A central ingredient of the algorithm \cite{DBLP:journals/jal/RobertsonS86} considers a set $S$
of size $O(k)$ goes through all partitions of $S$ into two balanced
subsets and for each such a partition checks existence of a small
separator. However, in our setting the set $S$ can contain arbitrarily many vertices, as long as it can be covered by a bounded number of hyperedges. The following statement provides us with an appropriate variant of the classic result for treewidth that is applicable to our setting.

\begin{restatable}{theorem}{RESTATEbalancedsep}
    \label{balancedsep}
Let $H$ be a hypergraph with $\ghw(H) \leq k$ and let $E' \subseteq E(H)$. Then there exists a weak partition of $E'$ into three sets $E'_0,E'_1,E'_2$ such that
    \begin{enumerate}
    \item there is a $(\bigcup E'_1,\, \bigcup E'_2)$-separator $S$ such that $\bigcup E'_0 \subseteq S$ and $\rho(S) \leq k$,
        \item $|E'_1| \leq \frac{2}{3} |E'|$, and $|E'_2| \leq \frac{2}{3} |E'|$.
    \end{enumerate}
\end{restatable}

\smallskip
Theorem \ref{balancedsep} allows us to consider all partitions of a small set of hyperedges covering
the given potentially large set of vertices thus guaranteeing an \fpt upper bound for the number
of such partitions. 

The other obstacle in upgrading the result \cite{DBLP:journals/jal/RobertsonS86}
is that a balanced separator is no longer required to be small but rather
to have a small edge cover number. In order to compute such a separator
we will need a witnessing tree decomposition of $H$ of a small \ghw. This is also the reason why we employ iterative compression rather than providing a direct algorithm for approximation.
However, even in presence of the tree decomposition, we were still unable
to design a 'neat' algorithm that would either produce an (approximately) small
separator or reject, implying that a small separator does not exist. 
Instead, we propose an algorithm for the following problem 
with a \emph{nuanced} reject that is still suitable
for our purposes. 

\begin{problem}{ApproxSep}
    Input & A $(2,d)$-hypergraph $H$, sets $A_1,A_2 \subseteq V(H)$, \\
    & TD $(T,\B)$ of $H$ with \ghw $p$, integers $0 \leq k_0 \leq k \leq p$ \\
    Parameters & $p$ and $d$ \\
    Output & An $(A_1,A_2)$-separator with edge cover number\\ & at most $(3k+d+1)(2k-1)k_0$, \\
    & or \textbf{Reject}, in which case there either exists no $(A_1,A_2)$-separator with edge cover number at most $k_0$ or $ghw(H)>k$. 
\end{problem}

\begin{restatable}{theorem}{RESTATEfptsep}
\label{fptsep}
    \textsc{ApproxSep} is fixed-parameter tractable.
\end{restatable}

We postpone to the next section a more detailed consideration
of the algorithm for \textsc{ApproxSep}. 
In the rest of this section we discuss the criterion for large \textsc{ghw}
used by the algorithm and the context in which the critetion is checked. 
For this purpose, 
we will require some technical definitions.
For hypergraph $H$, let us call $U \subseteq V(H)$ a \emph{subedge} (of $H$) if there is $e \in E(H)$ such that
$U \subseteq e$. We say that two subedges $U_1,U_2$ are \emph{incompatible} if their union $U_1 \cup U_2$ is not a subedge. 

\begin{definition}
  \label{def:shyg}
  An \emph{$(a,b)$ subedge hypergrid} (or $(a,b)$-shyg) consists of pairwise incompatible subedges $U_1, \dots, U_a, S_1,\dots, S_b$ such that:
  \begin{enumerate}
  \item $U_1,\dots,U_a$ are pairwise disjoint, and
  \item for each $j\in[b]$, $S_j$ touches $U_1,\dots,U_a$.
  \end{enumerate}
\end{definition}

Throughout this paper we will be interested in a specific dimension of subedge hypergrid. Namely for deciding width $k$ in $(2,d)$-hypergraphs, we will be interested in the existence of $(3k+d+1, \xi(k,d))$ subedge hypergrids, where $\xi$ is a function in $O((kd)^d)$ (refer to~\Cref{sec:combstat} for details). 
We will denote the set of all subedge grids of $H$ of this dimension by ${\bf S}_{k,d}(H)$. The reason we care about these subedge hypergrids in particular is that their existence is a sufficient condition for high $\ghw$.

\begin{restatable}{theorem}{RESTATEnogrid}
    \label{nogrid}
Let $H$ be a $(2,d)$-hypergraph such that ${\bf S}_{k,d}(H) \neq \emptyset$. Then $ghw(H)>k$.
\end{restatable}

In fact, the algorithm for the \textsc{ApproxSep} problem
either constructs a required separator or discovers that ${\bf S}_{k,d}(H) \neq \emptyset$. 
More precisely, the identification of an element of ${\bf S}_{k,d}(H)$ takes place within 
the procedure described in \Cref{setapprox} below.
The central part the algorithm for \textsc{ApproxSep}  
is to pick a vertex $t \in T$ (recall that $(T,{\bf B})$ is the input tree decomposition for $H$)
and then guess a set $W \subseteq {\bf B}(t)$ that shall be part of the separator 
being constructed. Technically, the guessing means a loop exploring a family
of subsets of ${\bf B}(t)$. This family must be of an \fpt size and the edge cover
of each element of the family must not be too large compared to $k_0$
This idea is formalised in the notion of \emph{gap cover approximator} formally defined below. 

\begin{restatable}{definition}{RESTATEgapcovdef}
  For hypergraph $H$ and set $U \subseteq V(H)$
  a \emph{$(\beta,\gamma)$-gap cover approximator} (for $U$) is a set $\mathbf{X} \subseteq 2^{V(H)}$
  such that
  \begin{romanenumerate}
  \item For each $X \in \mathbf{X}$, $\rho(X) \leq \beta$.
  \item For each $U' \subseteq U$ with $\rho(U') \leq \gamma$, there is a $X \in \mathbf{X}$ such that $U' \subseteq X$.
  \end{romanenumerate}
\end{restatable}

\begin{example}
  \label{exgca}
  \begin{figure}[t]
    \centering
    \begin{subfigure}[b]{0.35\textwidth}
      \centering
      \includegraphics[width=0.6\textwidth]{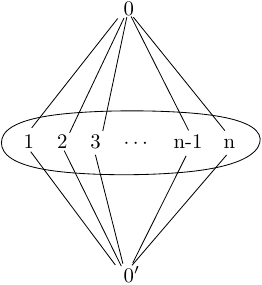}
      \caption{Hypergraph $H$}
      \label{fig:gcaex1}
    \end{subfigure}
    \hfill
    \begin{subfigure}[b]{0.6\textwidth}
      \centering
      \includegraphics[width=0.9\textwidth]{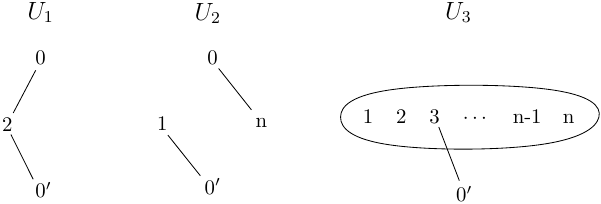}
      \caption{Examples of sets $U \subseteq V(H)$ with $\rho(U))=2$}
      \label{fig:gcaex2}
    \end{subfigure}
    \caption{Illustrations for \Cref{exgca}}
    \label{fig:gcaex}
  \end{figure}
  To illustrate the notion of gap cover approximators we consider the hypergraph $H$ in \Cref{fig:gcaex1} with edges $e_i=\{0,i\}$ and $e'_i=\{0', i\}$ for $i \in [n]$, and the edge $f = [n]$. \Cref{fig:gcaex2} illustrates some sets of vertices that can be covered by 2 edges of $H$. It is easy to see that the number of sets $U \subseteq V(H)$ with $\rho(U) \leq 2$ is roughly quadratic in $n$ (e.g., $\{0,0',i, j\}$ is such a set for every pair of $i,j \in [n]$). For our algorithmic goals such an exponential dependene on the cover number (here $2$) is problematic when we need to enumerate through all covers. Instead we can enumerate only over a gap cover approximator, for which we can potentially bound the size more tightly. In our example here, there is a single set with cover number $3$ that contains every $U$ with $\rho(U)\leq 2$, namely $V(H)$. That is $\mathbf{X}=\{V(H)\}$ is a $(3,2)$-gap cover approximator for $H$, which we could use in our algorithm to check certain properties only for $V(H)$, instead of checking for all $\Theta(n^2)$ many sets with cover number at most 2.
\end{example}

In order to produce the desired gap cover approximator, the 
algorithm solving the \textsc{ApproxSep} problem
runs a function $\gcaalg$. The function computes either
an gap cover approximator or an element of ${\bf S}_{k,d}(H)$  and, in the latter 
case,  rejects. In particular, when the algorithm rejects, we know (implicitly) that $\mathbf{S}_{k,d}(H) \neq \emptyset$, which in turn guarantees that $\ghw$ is greater than $k$ in this case and the rejection can be propagated to the top-level. A formal description of the behaviour of $\gcaalg$ is provided below.

%
%
\begin{restatable}{theorem}{RESTATEsetapprox}
    \label{setapprox}
There is an algorithm $\gcaalg(H,U,p,k,k_0)$ whose input is a $(2,d)$-hypergraph $H$, $U \subseteq V(H)$
with $\rho(U) \leq p$, and integers $k_0 \leq k \leq p$. 
The algorithm returns a
$((3k+d+1)k_0, k_0)$-gap cover approximator of $U$  
or \textbf{Reject}, in which case it is guaranteed that 
$\ghw(H)>k$. 
The algorithm is in \fpt when parameterised in $p$ and $d$. 
\end{restatable}

\section{Algorithmic Details} \label{sec:pseudo}
In this section we sketch proofs of Theorems \ref{thm:compress}
and \ref{fptsep}. In particular, we provide pseudocodes of the corresponding algorithms
and intuitive justification of their correctness and \fpt membership. 
Algorithm \ref{alg:compress} uses as a subroutine the algorithm
$AppSep$, which is discussed afterwards in \Cref{sec:appsepmain}.

\subsection{An \fpt Algorithm for the Compression Step (\Cref{thm:compress})}

To prove Theorem \ref{thm:compress}
we define Algorithm \ref{alg:compress}, prove that it is correct, and that the algorithm
works \fpt time.
In general principle the algorithm follows similar ideas to previous algorithms for checking \ghw (e.g.,~\cite{DBLP:journals/jea/FischlGLP21}) in that, at each stage, we separate the problem into subproblems for each connected component, and recurse. The set $W$ provides an interface to how the subproblem connects to the rest of the decomposition. By guaranteeing that $W$ is covered in the root of the decomposition for the subproblem (line 8), we guarantee that the decompositions for all the subproblems can be assembled into a decomposition for the parent call in lines 8 to 22. (see also~\cite{DBLP:conf/pods/GottlobLOP22} where a similar idea is formalised in terms of \emph{extended hypergraphs} in the context of checking plain hypertree width).

We move on to giving an overview of the argument for the runtime and correctness of the algorithm.
For the overall time complexity of the algorithm, we first observe that
a single recursive application of the algorithm runs in \fpt time.
The search for $E_W$ in line 1 is \fpt by using \Cref{fptsetcov} to find a cover for $W$ (the procedure from the proposition is constructive). If the produced cover is smaller than $3 \alpha(k,d)+1$ we can incrementally increase the size of the cover by searching for covers for $W' \supseteq W$ created by adding a vertex outside of the cover to $W$.
For line 2 we can naively iterate through all possible partitions of $E_W$ intro three sets and call $\mathit{AppSep}$, which itself is in \fpt by \Cref{fptsep}. For line 7 we note again that testing $\rho$ is \fpt by \Cref{fptsetcov}. From line 8 onward, except for the recursion, the algorithm performs straightforward manipulations of sets and hypergraphs that are of no deeper interest to our time bound.

\begin{algorithm}[t]
  \DontPrintSemicolon
  \SetKwInOut{Input}{Input}\SetKwInOut{Output}{Output}
  \Input{$(2,d)$-hypergraph $H$, a positive integer $k$, a TD $(T,{\bf B})$ of $H$
  with \ghw $4\alpha(k,d)+1$, $W \subseteq V(H)$ wih $\rho(W) \leq 3\alpha(k,d)$}
  
  $E_W \gets $ any $\rho$-stable subset of $E(H)$ covering $W$ of cardinality $3\alpha(k,d)+1$\;
  Find a weak partition $E_0,E_1,E_2$ of $E_W$ s.t. $\rho(E_0)\leq k$,  $\rho(E_1),\rho(E_2) \leq 2\alpha(k,d)$, and $\mathit{AppSep}(H, \bigcup E_1, \bigcup E_2, (T,\B), k, k, p, V(T))$ does not return \textbf{Reject}\;
  \eIf{no such weak partition exists}{
    \KwRet{\textbf{Reject}}\;
  }{
    $X \gets \mathit{AppSep}(H, \bigcup E_1, \bigcup E_2, (T, \B),k, k, p, V(T))$\;
  }
  $U_1,\dots, U_q \gets \{C \in \mathit{CComps}(H\setminus X) \mid \rho(C \cup X) > 4\alpha(k,d)\}$\;
  Let $T^*$ be a tree with a new node $r$ and $B^*(r)=W \cup X$\;
  \For{$i \in [q]$}{
    $H_i \gets H[U_i \cup X]$\;
    $W_i \gets (\bigcup E_W \cap U_i) \cup X$\;
    ${\bf B}_i \gets \{t \mapsto {\bf B}(t) \cap (U_i \cup X) \mid t \in V(T) \}$\;
    $O_i \gets \textit{Compress}(H_i, k, (T,{\bf B}_i), W_i)$\; \label{line:compress.rec}
    \eIf{$O_i$ is \textbf{Reject}}{
      \KwRet{\textbf{Reject}}\;
    }{
      $(T'_i, {\bf B'}_i) \gets O_i$\;
      $t_i \gets $ a node $u$ of $T'_i$ s.t. $W_i \subseteq {\bf B'}_i(u)$\;
      Add $T'_i$ to $T^*$ by making $t_i$ a neighbour of $r$ and let ${\bf B^*}(u) = {\bf B'}_i(u)$ for all $u \in T'_i$.
    }
  }
  \For{$U \in \mathit{CComps}(H \setminus X)$ where $\rho(U \cup X) \leq 4 \alpha(k,d)$}{
    Add new node $u$ as a neighbour of $r$ to $T^*$.\;
    Set ${\bf B^*}(u) = U \cup X$.\;
  }
  \KwRet{$(T^*, {\bf B^*})$}
  \caption{The algorithm $\mathit{Compress}(H, k, (T,\B), W)$.}
  \label{alg:compress}
\end{algorithm}

Next, we observe that the number
of recursive applications is, in fact, polynomial in $H$. 
We naturally organise recursive applications into a successors (recursion) tree and upper 
bound the number of nodes of the tree by the product of the height of the tree 
and the number of leaves. We observe that the height of the tree is at most $|V(H)|$. 
For this we prove two auxiliary
statements. The first is that for $X$, as computed in line 6, $H \setminus X$ 
has at least two connected components. The second, immediately following from the first
one is that for each $H_i$, created in line 10, $|V(H_i)|<|V(H)|$. 
Thus it follows that the number of vertices of the input hypergraph decreases
as we go down the recursion tree thus implying the upper bound on the height
of the tree. 

Additionally, we prove that the number of leaves is no larger than $\rho(H)^2$. 
The main part of this proof is an induction for the case where the number
of sets $U_1, \dots, U_q$ obtained at line 7 is at least $2$.
In particular, we notice that since 
$\rho(U_i) \geq 3\rho(X)$ for each $i \in[q]$ and since 
$\rho(\bigcup_{i=1}^q U_i)=\sum_{i=1}^q \rho(U_i)$, it holds that 
\[
\rho(U_i \cup X)^2 \leq (\rho(U_i)+\rho(X))^2 \leq (\sum_{i=1}^q \rho(U_i))^2=
 \rho(\bigcup_{i=1}^q U_i)^2 \leq \rho(H)^2.
 \]

To prove correctness of the \Reject output, 
we observe that return of \Reject by the whole algorithm is triggered 
by return of \Reject on line 4, failure to find an appropriate weak partition, or by rejection in one of its recursive applications. 
By \Cref{balancedsep}
and the $\rho$-stability of $E'$,
the $\mathbf{Reject}$ on line 4 implies that either $H$, or one of its induced
subgraphs have \ghw greater than $k$. In the latter case, of course also $\ghw(H) > k$.

Finally, the two main aspects of correctness of the non-rejection 
output are the upper bound on the \ghw of the resulting
tree decomposition and that 
the properties of the tree decomposition are 
not lost by the 'gluing' procedure
as specified in lines 8-22 of the algorithm. 
The requirement that $E_W$ must cover $W$ is essential for ensuring that the properties
of the tree decomposition are not destroyed by the gluing. Intuitively, the parameter $W_i$ in the recursion on Line~\ref{line:compress.rec} represents the connection of the component $H_i$ with the rest of the decomposition. 

 \subsection{Finding Approximate Separators in \fpt (\Cref{fptsep})}
 \label{sec:appsepmain}
The proof of Theorem \ref{fptsep} requires significant extension of notation.
First, for a tree decomposition $(T,{\bf B})$ of a hypergraph and $X \subseteq V(T)$,
we denote by $ghw(T,{\bf B},X)$ 
the maximum of $\rho({\bf B}(t))$ among $t \in X$. 
We are looking for a separator subject to several constraints.
Repeating these constraints every time we refer to a separator
is somewhat distracting and we therefore define the set
of separators that we need to consider for this overview. 
Define $\sep(H,A,B,k_0,(T,{\bf B}),X)$
as the set of all $(A,B)$-separators $W$ of $H$ with $\rho(W) \leq k_0$
and $W \subseteq {\bf B}(X)$ where $(T,{\bf B})$ is a tree decomposition
of $H$ and $X \subseteq V(T)$.  

The following theorem is a generalisation of Theorem \ref{fptsep}
\begin{restatable}{theorem}{RESTATEgenapp}
    \label{genfptsep}
There is an algorithm $AppSep(H,A,B,(T,{\bf B}),k_0,k,p,X)$
whose input is a $(2,d)$-hypergraph $H$, $A,B \subseteq V(H)$,
three positive integers $k_0 \leq k \leq p$, a tree decomposition
$(T,{\bf B})$ of $H$ and $X \subseteq V(T)$ such that all the elements
of $V(T) \setminus X$ are leaves and $ghw(T,{\bf B},X) \leq p$. 
The algorithm either returns an element 
of $sep(H,A,B,(3k+d+1)(2k-1)k_0,(T,{\bf B}),X)$ or \textbf{Reject}.
In the latter case, it is guaranteed that either 
$ghw(H)>k$ or $sep(H,A,B,k_0,(T,{\bf B}),X)=\emptyset$
\end{restatable}

Clearly, \Cref{genfptsep} implies \Cref{fptsep}
by setting $X=V(T)$. The reason we need this extra parameter is
that in the recursive applications of $\mathit{AppSep}$ some bags may
have the edge cover number larger than $p$, so we keep track of the set
of nodes whose bags are 'small'. 

To present the pseudocode, we define a specific choice of a subtree of the given tree.
Let $T$ be a tree, $t \in V(T)$, $Y \subset V(T)$.
Then $T_{t,Y}$ is the subtree of $T$ which is the
union of all paths starting from $t$ whose second
vertex belongs to $Y$.
The notion of $T_{t,Y}$ naturally extends to subsets
of $V(T)$ and to tree decompositions where $T$
is the underlying tree. In particular, for $X \subseteq V(T)$,
we denote $X \cap V(T_{t,Y})$ by $X_{t,Y}$. 
Next, if $(T,{\bf B})$ is a tree decomposition of $H$
then by denote by ${\bf B}_{t,Y}$ the restriction to ${\bf B}$
to $V(T_{t,Y})$ and by $H_{t,Y}$ the graph $H[{\bf B}(V(T_{t,Y}))]$. 

We also introduce a variant $T^+_{t,Y}$ of $T_{t,Y}$ which will be needed
for recursive applications of $AppSep$.
The tree $T^+_{t,Y}$ is obtained from $T_{t,Y}$ by introducing a new 
node $r$ and making it adjacent to $t$. 
The function ${\bf B}^+_{t,Y}$ is obtained from ${\bf B}_{t,Y}$
by setting ${\bf B}^+_{t,Y}(r)=\bigcup_{t' \in V(T) \setminus V(T_{t,Y})} {\bf B}(t')$. 
One final notational convention concerns 
adjusting a tree decomposition $(T,{\bf B})$ of $H$
in case a set $W \subseteq V(H)$ is removed from $H$.
In this case we set ${\bf B}^{-W}(t')={\bf B}(t') \setminus W$
for each $t' \in V(T)$. 

The following statement is important for verifying that these recursive applications are well-formed.

\begin{restatable}{theorem}{RESTATEvalidtw}
    \label{validtw}
Let $H$ be a hypergraph, $(T,{\bf B})$ a TD of $H$,
$t \in V(T)$, $Y \subseteq N_T(t)$. 
Then $(T_{t,Y},{\bf B}_{t,Y})$ is a TD of $H_{t,Y}$,
$(T^+_{t,Y},{\bf B}^+_{t,Y})$ is a TD of $H$ and
$(T,{\bf B}^{-W})$ is a TD of $H \setminus W$. 
Moreover, let $X \subseteq V(T)$ such that 
all the vertices of $V(T) \setminus X$ are leaves of $T$. 
Then all the vertices of $V(T_{t,Y}) \setminus X_{t,Y}$
are leaves of $T_{t,Y}$. 
\end{restatable}

The algorithm (roughly speaking) chooses 
a vertex $t \in V(T)$, partitions $N(t)$ into $Y_1$
and $Y_2$ and applies recursively to 
$H_{t,Y_1}$ and $H_{t,Y_2}$. However, the triple
$(t,Y_1,Y_2)$ is chosen not arbitrarily but in
a way that both $X_{t,Y_1}$ and $X_{t,Y_2}$ are
significantly smaller than $X$. The possibility of such a 
choice is guaranteed by the following theorem. 

\begin{restatable}{theorem}{RESTATEbalancedvert}
    \label{balancedvert}
Let $T$ be a tree, $X \subseteq  V(T)$ such that all
$|X| \geq 3$ and all the vertices of $V(T) \setminus X$
are leaves. 
Then there is $t \in X$ with $deg_{T[X]}(t) \geq 2$
and a partition $Y_1,Y_2$ on $N_T(t)$ so that
for each $i \in \{1,2\}$, $|X_{t,Y_i}| \leq 3/4|X|$. 
Moreover, the triple $(t,Y_1,Y_2)$ can be computed 
in a polynomial time. 
\end{restatable}

Theorem \ref{balancedvert} can be seen as a variant of a classical
statement that a rooted tree has a descendant rooting a subtree
with the number of leaves between one third to two third of the total
number of leaves. The proof is based on a similar argument of picking
a root and gradually descending towards a 'large' subtree until the desired
triple is found. 

For the validity of $AppSep$, it is important to note that
each $X_{t,Y_i}$ preserves for $T_{t,Y_i}$ the invariant that
all the $V(T_{t,Y_i}) \setminus X_{t,Y_i}$ are leaves of $T_{t,Y_i}$.  
We are almost ready to consider the pseudocode, it only remains 
to identify auxiliary functions. 
In particular $\mathit{GetBalVert}(T,X)$ is a polynomial time algorithm
as specified in Theorem \ref{balancedvert}. 
Also recall that $\gcaalg$ is an \fpt algorithm constructing 
a $((3k+d+1)k_0, k_0)$-gap cover approximator in the way specified by Theorem \ref{setapprox}.

The pseudocode of $AppSep$ is presented in 
Algorithm \ref{alg:main}. 
For the sake of readability, we make two notational conventions.
First, since parameters $p$ and $k$ do not change when passed through
recursive calls, we consider them fixed and do not mention them as part of the input when recursing. 
Second, we move consideration of the case with $|X| \leq 2$ into a separate
function $\mathit{SmallSep}$ provided in Algorithm~\ref{alg:smallsep} and use it
as an auxiliary function in Algorithm~\ref{alg:sepgen}. Here the idea is straightforward, we naively test for all gap cover approximators whether they are $(A,B)$-separators.

\begin{algorithm}[t]
    \DontPrintSemicolon
    \SetKwInOut{Input}{Input}\SetKwInOut{Output}{Output}
    \Input{$(2,d)$-hypergraph $H$, $A,B \subseteq V(H)$, positive integers $k_0 \leq k \leq p$,
    a TD $(T,{\bf B})$ of $H$, $X \subseteq V(T)$, $|X| \leq 2$ all the elements of $V(T) \setminus X$ 
    are leaves}
     \eIf{$|X|=1$}{
     $\{t\} \gets X$\;
     $Sets \gets \mathit{\gcaalg}(H,{\bf B}(t),p,k,k_0)$\;
  }{
    $\{t_1,t_2\} \gets X$\;
     $Sets \gets \mathit{\gcaalg}(H,{\bf B}(t_1) \cup {\bf B}(t_2),p,k,k_0)$\;
  }
  \If{$Sets$ is \textbf{Reject}}{
     \KwRet{\textbf{Reject}}\;
  }
  \For{$W \in Sets$}{
      \If{$W$ is an $(A,B)$-separator}{
        \KwRet{$U$}\;
      }
    }
  \KwRet{\textbf{Reject}}\;
   \caption{The algorithm $\mathit{SmallSep}(H,A, B,p,k,k_0,(T,{\bf B}),X)$}
    \label{alg:smallsep}
  \end{algorithm}

\begin{algorithm}[t]
    \DontPrintSemicolon
    \SetKwInOut{Input}{Input}\SetKwInOut{Output}{Output}
    \Input{$(2,d)$-hypergraph $H$, $A,B \subseteq V(H)$, 
    a TD $(T,{ B})$ of $H$, positive integers $0 \leq k_0 \leq k \leq p$, $X \subseteq V(T)$ s.t. all the elements of $V(T) \setminus X$ 
    are leaves}
    \If{$|X| \leq 2$}{
        \KwRet{$\mathit{SmallSep}(H,A, B, p,k, k_0,(T,{\bf B}),X)$}\;
      }
    $(t,Y_1,Y_2) \gets \mathit{GetBalVert}(T,X)$\;
    $Out \gets \mathit{AppSep}(H,A,B,(T_{t,Y_1}^{+},  {\bf B}_{t,Y_1}^{+}), k_0, X_{t,Y_1})$\;
    \If{$Out$ is not \textbf{Reject}}{
        \KwRet{$Out$}\;
      }
    $Out \gets \mathit{AppSep}(H,A,B,(T_{t,Y_2}^{+}, {\bf B}_{t,Y_2}^{+}), k_0, X_{t,Y_2})$\;
    \If{$Out$ is not \textbf{Reject}}{
        \KwRet{$Out$}\;
      }
    $Sets \gets \mathit{\gcaalg}(H,{\bf B}(t),p,k,k_0)$\;
     \If{$Sets$ is \textbf{Reject}}{
          \KwRet{\textbf{Reject}}\;
        }

     \For{each $W \in Sets$, each $k_1,k_2>0$ s.t. $k_1+k_2 \leq k_0$, and
             each weak partition $C_1,C_2$ of ${\bf B}(t) \setminus W$
               into unions of connected components of $H[{\bf B}(t) \setminus W]$}{

       $Out_1 \gets \mathit{AppSep}(H_{t,Y_1} \setminus W,
      (A_{t,Y_1} \cup C_1) \setminus W,(B_{t,Y_1} \cup C_1) \setminus W,
      (T_{t,Y_1},{\bf B}_{t,Y_1}^{-W}), k_1, 
      X_{t,Y_1})$\;

      $Out_2 \gets \mathit{AppSep}(H_{t,Y_2} \setminus W,
      (A_{t,Y_2} \cup C_2) \setminus W,(B_{t,Y_2} \cup C_2) \setminus W,
      (T_{t,Y_2},{\bf B}_{t,Y_2}^{-W}), k_2,
      X_{t,Y_2})$\;
      \If{neither of $Out_1$, $Out_2$ is \textbf{Reject}}{
        \KwRet{$Out_1 \cup Out_2 \cup W$}\;
      }
    }
   \KwRet{\textbf{Reject}}\;

    \caption{The algorithm $\mathit{AppSep}(H,A,B,(T,{\bf B}),k_0,k,p, X)$.}
    \label{alg:sepgen}
  \end{algorithm}

The general case in $\mathit{AppSep}$ is considerably more complex. Intuitively, the algorithm searches for separators in small local parts of the tree decomposition by searching through the output of \gcaalg called on the component of that local part of the decomposition (Lines 1 and 2 of the algorithm).
In general this is of course not sufficient to find $(A,B)$-separators. What we do instead is to try and find partial separators that ultimately combine into a single $(A,B)$-separator. To that end we split the tree decomposition into two decompositions in a balanced fashion (Line 3). The two cases handled from Lines 4 to 9 cover the special case where the search is propagated to one part of the tree decomposition, while the body of the loop at Line 13 considers (roughly speaking) all possibilities of splitting up the separator over both parts.

The first step of proving  
Theorem \ref{genfptsep} is to prove the \fpt runtime
of $AppSep$. We first observe that a single application of
$AppSep$ takes \fpt time. 
The efficiency of $\mathit{GetBalVert}$ and $\mathit{\gcaalg}$ has been discussed
above. $SmallSep$ is effectively a loop over the output of $\mathit{\gcaalg}$
with polynomial time spent per element. Note that since \gcaalg is computed in \fpt time we also obtain a corresponding bound on the size  of the $((3k+d+1)k_0, k_0)$-gap cover approximators to iterate over. Finally, in the loop of Line 13, we
only need to observe that the number of connected components
of ${\bf B}(t) \setminus W$ is at most $p$ as otherwise the edge cover
number of ${\bf B}(t)$ is greater than $p$. Hence, the number of partitions
$C_1,C_2$ considered in the loop is $O(2^p)$. 

Next, we need to demonstrate that the number of recursive applications
of $\mathit{AppSep}$ is \fpt. We present the number of applications
as a recursive function $F(k_0,m)$ where $m=|X|$. If $m \leq 2$
then there is only a single application through running $\mathit{SmallSep}$.
Otherwise, there is one recursive application at Line 4 and one at Line 7 
where the first parameters remains the same and the second parameter is at most $3/4m$
(by selection of $t,Y_1,Y_2$). 
Additionally, there are also the recursive applications   
in Lines 14 and 15 where the first parameter is at most $k_0-1$ and the second parameter is at most $3/4m$. 
As a result, we obtain a recursive formula $F(k_0,m) \leq 2F(k_0,3/4m)+g(p)F(k_0-1,3/4m)$. 
We note that this function can be bounded above by a fixed-parameter cubic function (see \Cref{numrecur} in the appendix for details)
thus establishing the \fpt runtime of $\mathit{AppSep}$. 

Next, we need to demonstrate correctness of the non-rejection output. 
It is straightforward to see by construction that the returned set $S$ is a subset
of ${\bf B}(X)$: ultimately, the set $S$ is comprised of unions of outputs of $\gcaalg(H, U, \dots)$, either directly in Line~10, or indirectly via $\mathit{SmallSep}$. In the first case, the returned sets are a subset of $\mathbf{B}(t)$, where $t \in X$ by definition of $\mathit{GetBalVert}$. In the second case, $U \subseteq \mathbf{B}(X_{t,Y_i})$ for $i=1$ or $i=2$ by definition of $\mathit{SmallSep}$. By definition, both $X_{t,Y_i}$ are subsets of $X$ and thus $\mathbf{B}(X_{t,Y_i})$ is a subset of $\mathbf{B}(X)$. Since the recursion always restricts parameter $X$ to either $X_{t,Y_1}$ or $X_{t,Y_2}$ the inductive application of this observation is immediate.

We need to show $S$ is an $(A,B)$-separator and that its edge cover number is within a specified upper bound. 
Both claims are established by induction. The main part of proving that $S$ is an
$(A,B)$-separator is showing that if $S$ as returned on Line 17 
then it is an $(A,B)$-separator. This follows from the induction assumption applied 
to $Out_1$ and $Out_2$ and the following statement. 

\begin{restatable}{lemma}{RESTATEcomposedsep}
    \label{composedsep}
Let $H$ be a hypergraph, $V_1,V_2 \subseteq V(H)$ be such that
$V_1 \cup V_2=V(H)$ and $Y=V_1 \cap V_2$ is a $(V_1,V_2)$-separator. 
Let $W \subseteq Y$ and let $C_1,C_2$ be a weak partition of $Y \setminus W$. 
Let $H_1=H[V_1 \setminus W]$ and $H_2=H[V_2 \setminus W]$.
Let $A,B \subseteq V(H)$. For each $i \in \{1,2\}$
let $A_i=(A \cap V(H_i)) \cup C_1$, let $B_i=(B \cap V(H_i)) \cup C_2$, and let $W_i$ be an $(A_i,B_i)$-separator of $H_i$.
Then $W_1 \cup W_2 \cup W$ is an $(A,B)$-separator of $H$. 
\end{restatable}

To prove that the size of the output $S$ of $\mathit{AppSep}$ matches the
required upper bound, we observe that the output is
the union of several sets $S_1, \dots, S_q$ each of which is in a family
returned by an application of $\gcaalg$. 
This guarantees that $\rho(S_i) \leq(3k+d+1)k_0$. 
It only remains to show that $q \leq 2k_0-1$ ($q$ is the number of sets $S_1,\dots,S_q$ whose union make up $S$). 
To this end we observe that the recursive applications invoking
$\gcaalg$ can be naturally organised into a recursion tree
where each node has two children, accounting for the recursive applications 
in Lines 14 and 15  (the applications on Lines 4 and 7 are not relevant since there  is no invocation of $\gcaalg$ associated with them). 
Then $q$ is simply the number of nodes  of the tree. 
By a simple induction we observe that if $k_1$ and $k_2$ are
the numbers as obtained in Line 13 then $k_1+k_2 \leq k_0$
and the number of nodes rooted by children of the tree is 
at most $2k_1-1$ and $2k_2-1$, respectively.
Hence, the total number of nodes, accounting for the root, 
is at most $(2k_1-1)+(2k_2-1)+1 \leq 2(k_1+k_2)-1 \leq 2k_0-1$ as required. 

Let us now discuss correctness of the \textbf{Reject} output.
For this, we first recursively define the \textbf{Reject} \emph{triggered} by 
$\gcaalg$. This happens when $\mathit{AppSep}$ runs $\mathit{SmallSep}$ and the latter returns 
\textbf{Reject} in Lines 5 or 10 of Algorithm~\ref{alg:smallsep}, or
\textbf{Reject} is returned in Line 12 of Algorithm~\ref{alg:sepgen},
or when \textbf{Reject} is returned in Line 18 of Algorithm~\ref{alg:sepgen}
and one of recursive applications leading to this output returns 
\textbf{Reject} triggered by $\gcaalg$. By inductive application of 
 \Cref{setapprox} we observe that \textbf{Reject} triggered by $\gcaalg$
implies that the \ghw of some induced subgraph of $H$ 
(and hence of $H$ itself) is greater than $k$. 

Finally, we demonstrate that if the \textbf{Reject} output is not
triggered by $\gcaalg$ then $\sep(H,A,B,k_0,(T,{\bf B}),X)=\emptyset$. 
First, we assume that $|X| \leq 2$ and demonstrate this for 
Algorithm~\ref{alg:smallsep}. 
It follows from the description 
that, in the considered case, no element of $Sets$ is an
$(A,B)$ separator. As each $W' \subseteq {\bf B}(X)$ with
$\rho(W') \leq k_0$ is a subset of some element of $Sets$,
it follows that no such $W'$ is an $(A,B)$-separator of $H$. 
In the case where $|X| \geq 3$, a \textbf{Reject} not inherited from
$\gcaalg$ can only be returned in Line 18 of Algorithm~\ref{alg:main}. 
This, in particular, requires \textbf{Reject} to be returned by recursive
applications in Lines 4 or 7. 
By the induction assumption,
$\sep(H,A,B,k_0,(T_{t,Y_i}^{+},{\bf B}_{t,Y_i}^{+}),X_{t,Y_i})=\emptyset$ for each $i \in \{1,2\}$.
This means that if there is $W^* \in \sep(H,A,B,k_0,(T,{\bf B}),X)$
then $W^*$ is \emph{not} a subset of ${\bf B}(X_{t,Y_1})$ nor 
of ${\bf B}(X_{t,Y_2})$. We conclude that by the induction assumption,
such a $W^*$ would cause a non-rejection output in one of
iterations of the loop in Line 13. Since the algorithm passes through to Line 18,
such an iteration does not happen so we conclude that such a 
$W^*$ does not exist. 

\section{Combinatorial Statements} \label{sec:combstat}
In this section we prove \Cref{balancedsep} and sketch the proofs for Theorems \ref{nogrid} and \ref{setapprox}. While \Cref{setapprox} also refers to the existence of an algorithm, we consider the nature of the theorem to be purely combinatorial. The resulting algorithm is simply a naive enumeration of all possibilities of combining certain sets. 

\subsection{\Cref{balancedsep}}
For \Cref{balancedsep} we can make use of a result from the literature and prove the statement in full here. We first recall key terminology from Adler et al.~\cite{DBLP:journals/ejc/AdlerGG07}, who proved the result that we will use. 
For a set $E' \subseteq E(H)$, and $C \subseteq V(H)$ define $\mathit{ext}(C,E') := \{e \in E' \mid e \cap C \neq \emptyset \}$. We say that $C$ is \emph{$E'$-big} if $|\mathit{ext}(C,E')| > \frac{|E'|}{2}$. A set $E' \subseteq E(H)$ is \emph{$k$-hyperlinked} if for every set $S \subseteq E(H)$ with $|S|< k$, $H \setminus \bigcup S$ has an $E'$-big connected component. The hyperlinkedness $\mathit{hlink}(H)$ of $H$ is the maximal $k$ such that $H$ contains a $k$-hyperlinked set. From another perspective, if $\mathit{hlink}(H) \leq k$, then for any set $E' \subseteq E(H)$, there is an $S\subseteq E(H)$ with $|S|\leq k$ such that no connected component of $H\setminus \bigcup S$ is $E'$-big.
Adler et al.~\cite{DBLP:journals/ejc/AdlerGG07} showed the following.
\begin{proposition}[\cite{DBLP:journals/ejc/AdlerGG07}]
\label{hlink}
For every hypergraph $H$, $\mathit{hlink}(H) \leq \ghw(H)$.
\end{proposition}

\begin{proof}[Proof of \Cref{balancedsep}]
    By assumption and \Cref{hlink}, we have $\mathit{hlink}(H)\leq k$. Then for $E'$, there is a set of $k$ edges $S \subseteq E(H)$ such that no connected component of $H \setminus S$ is $E'$-big.
    Let $C_1,\dots,C_\ell$ be the connected components of $H \setminus \bigcup S$. First, observe that $\mathit{ext}(C_i,E')\cap \mathit{ext}(C_j,E') = \emptyset$ for any distinct $i,j\in [\ell]$.
    Suppose, w.l.o.g., that $|\mathit{ext}(C_i,E')| \geq |\mathit{ext}(C_{i+1},E')|$ for $i \in [\ell-1]$. Let $m$
    be the highest integer such that $\sum_{i=1}^m |\mathit{ext}(C_i,E')| < \frac{2}{3}|E'|$. Such an $m\geq 1$ always exists because no component is $E'$-big.
    We claim that $E'_0 = S \cap E'$, $E'_1 = \bigcup_{i=1}^m \mathit{ext}(C_i,E')$, and $E'_2 = \bigcup_{i=m+1}^\ell \mathit{ext}(C_i,E')$ are as desired by the statement.

    It is clear that $E'_0$ satisfies the condition of the lemma (for separator $\bigcup S$).
    Furthermore, $\bigcup S$  is an $\left(\bigcup E'_1,\bigcup E'_2 \right)$-separator as the two sets touch unions of different $H \setminus \bigcup S$ components. The size bound $|E'_1| \leq \frac{2}{3}|E'|$ holds by construction. 

    What is left to show is that the size bound also holds for $E'_2$. To that end we first observe that $|E'_1| \geq \frac{1}{3}|E'|$. Indeed, by the ordering of components by the size of $\mathit{ext}$, we have that $\mathit{ext}(C_{m+1},E') \leq |E'_1|$. Thus, $|E'_1| < \frac{1}{3}|E'|$ would contradict the choice of $m$. Since $E'_1$ and $E'_2$ are disjoint, this leaves at most $|E'|-|E'_1| \leq \frac{2}{3}|E'|$ edges for $E'_2$.
\end{proof}

\Cref{hlink} already plays an important role in the state of the art of \ghw computation. The implication of a so-called \emph{balanced separator} of size $k$ has been a key ingredient for various practical implementations for computing \ghw and related parameters~\cite{DBLP:conf/pods/FischlGP18,DBLP:journals/constraints/GottlobOP22,DBLP:conf/pods/GottlobLOP22}.
Our application of this idea is somewhat different from this prior work. There, balanced separators are used to reduce the search space of separators that need to be checked, and to split up the problem into small subproblems. Our applications of \Cref{balancedsep} is different: in \Cref{alg:compress} for \textsc{Compress} we use it to find a way to separate the interface to the parent node in the decomposition. Notably, this search requires the split into only a constant number of sets (rather than the possibly linear number of connected components) which is part of why we require our variation of the previous hyperlinkedness result. 

\subsection{Large Subedge Hypergrids (\Cref{nogrid})}
Recall from \Cref{def:shyg} that an \emph{$(a,b)$ subedge hypergrid} (or $(a,b)$-shyg) consists of pairwise incompatible subedges $U_1, \dots, U_a, S_1,\dots, S_b$ such that:
 $U_1,\dots,U_a$ are pairwise disjoint, and for each $j\in[b]$, $S_j$ touches $U_1,\dots,U_a$.
  Our proof of Theorem \ref{nogrid} first relates $(a,b)$-shygs to more restricted structures that we call \emph{strong} $(a,b)$-shygs.

\begin{definition} \label{strongshyg}
An $(a,b)$-shyg $U_1, \dots, U_a,S_1, \dots, S_b$ 
is \emph{strong}
if $S_j \cap S_{j'} \cap  \bigcup_{i \in [a]} U_i=\emptyset$
for each $j \neq j' \in [b]$. 
\end{definition}

\begin{restatable}{theorem}{RESTATEstrongnogrid}
    \label{strongnogrid}
Let $H$ be a $(2,d)$ hypergraph having a strong $(3k+1,(3k+1)d+1)$-shyg.
Then $ghw(H)>k$. 
\end{restatable}

To prove Theorem \ref{nogrid}, we first prove Theorem \ref{strongnogrid}
and then demonstrate that non-emptiness of ${\bf S}_{k,d}$ implies
existence of a strong $(3k+1,(3k+1)d+1)$-shyg. We continue with an overview of our proof of \Cref{strongnogrid}.

An important observation for subedges in $(2,d)$-hypergraphs is that if a subedge $U$ is large enough, and in particular if $|U|>d$, then this will uniquely determine the edge $e$ such that $U \subseteq e$. In this section we will refer to this uniquely determined $e$ as $e(U)$. Using this we state the following auxiliary lemma that gives us a lower bound for separating two subedges that are part of a shyg.

\begin{restatable}{lemma}{RESTATEtwostrong}
    \label{twostrong}
Let $U_1,U_2,S_1, \dots, S_b$ be a strong $(2,b)$-shyg of a $(2,d)$-hypergraph 
$H$. Let $\{e_1, \dots, e_q\} \subseteq E(H)$ be such that
$\{e(U_1),e(U_2)\} \cap \{e_1, \dots, e_q\}=\emptyset$ 
and $W=\bigcup_{i \in [q]} e_i$ is a $U_1,U_2$-separator.
Then $q \geq b/2d$. 
\end{restatable}
\begin{proof}[Proof Sketch]
    To separate $U_1$ and $U_2$, 
$W$ needs to contain either $S_i \cap U_1$
or $S_i \cap U_2$ for each $i \in [b]$.
For each $j \in [q]$, $e_j \cap U_1$ can intersect at most $d$
of $S_i$ and the same for $e_j \cap U_2$. 
Consequently, the union of $e_1, \dots, e_q$
can intersect at most $2qd$ sets of 
the $S_i \cap U_1$ or of $S_i \cap U_2$ and this
clearly does not exceed the number of the sets 
contained in the union. 
\end{proof}

Back to the proof of \Cref{strongnogrid},
we observe that $e(U_1), \dots, e(U_{3k+1})$
is $\rho$-stable. This is because each $e(U_i)$,
to 'incorporate' all $S_j$ must be of size at least $(3k+1)d+1$
and, in a $(2,d)$-hypergraph, this is too large to be covered
by $(3k+1)$ other hyperedges. We will then use Theorem \ref{balancedsep}
to prove the desired lower bound on $ghw(H)$, by showing that there is the set $\{e(U_1), \dots, e(U_{3k+1}\}$ cannot be separated in the way specified by \Cref{balancedsep}.
Towards a contradiction, we assume existence of a weak
partition $E'_0,E'_1,E'_2$ of $\{e(U_1), \dots, e(U_{3k+1}\}$, $|E'_i| \leq 2k$ for each $i \in \{1,2\}$
and $W \subseteq V(H)$ such that $\rho(W) \leq k$,
$\bigcup E'_0 \subseteq W$ and $W$ is a $\bigcup E'_1, \bigcup E'_2$-
separator. W.l.o.g. we assume existence of $\{e_1, \dots, e_r\} \subseteq E(H)$, $r \leq k$
such that $W=\bigcup_{i \in [q]} e_i$. 
Let $E^*_0$ be the set of all elements of $e(U_1), \dots, e(U_{3k+1})$ that are subsets
of $W$. We note that $E'_0 \subseteq E^*_0$ and, since no $e(U_i)$ can be covered by $k$ other hyperedges, $E^*_0 \subseteq \{e_1, \dots, e_r\}$. 
We also note that both $E'_1 \setminus E^*_0$ and $E'_2 \setminus E^*_0$
are nonempty. Indeed, if say $E'_2 \setminus E^*_0=\emptyset$ then
$3k+1=|\{e(U_1), \dots, e(U_{3k+1})\}|=|E'_1 \cup E^*_0| \leq 2k+k$, or in other words, in such a situation $E'_0,E'_1,E'_2$ could not form a weak partition of $3k+1$ edges.
Let $e(U_{i_1})  \in E'_1 \setminus E^*_0$ and $e(U_{i_2})) \in E'_2 \setminus E^*_0$. 
Then $W$ is a $(U_{i_1},U_{i_2})$-separator. By Lemma \ref{twostrong},
$r \geq (3k+1)d/2d>k$, and we arrive at a contradiction. This completes the sketch of the proof for \Cref{strongnogrid}.

As a next step, we show that a large enough shyg will imply the existence of a strong $(3k+1,(3k+1)d+1)$-shyg. 
We will show this inductively via a graded version of strong shygs that we will call \emph{$c$-strong shygs}. The only difference from
Definition \ref{strongshyg}
is that $|S_j \cap S_{j'} \cap  \bigcup_{i \in [a]} U_i| \leq c$
for each $j \neq j' \in [b]$. 
Thus a strong shyg is $0$-strong one and any ordinary 
shyg is $d$-strong by definition of a $(2,d)$-hypergraph.

Let us recursively define a function $g(c)=g_{k,d}(c)$
as follows. Let $g(0)=(3k+1)d+1$. 
For $c>0$, assuming that $g(c-1)$ has been defined,
we let $g(c)=g(0)^2 (g(c-1)-2))+1$. 
Further on, we let $\xi(k,d)=g_{k,d}(d)$
and let ${\bf S}_{k,d}(H)$ to be the set of all
$(3k+d+1,\xi(k,d))$-shygs of $H$. 
The second part of the proof of Theorem \ref{nogrid}
is the following statement.  
 
\begin{restatable}{theorem}{RESTATEcontstrong}
    \label{contstrong}
Let $c \geq 0$ and 
let $H$ be a $(2,d)$-hypergraph that has a $c$-strong $(3k+1+c,g(c))$-shyg.
Then $H$ has a strong $(3k+1,(3k+1)d+1)$-shyg. 
\end{restatable}

Theorem \ref{nogrid} is immediate from the combination 
of Theorem \ref{contstrong} with $c=d$ and Theorem \ref{strongnogrid}.
So, let us discuss the proof of Theorem \ref{contstrong}. 
\begin{proof}[Proof Sketch]
The proof is by induction on $c$.
The case $c=0$ is immediate as the considered
shyg is exactly the desired strong shyg. 
For $c>0$, we demonstrate we can 'extract' 
from the considered shyg either a
$c-1$-strong $(3k+c,g(c-1))$ shyg (implying 
the theorem by the induction assumption)
a strong $(3k+1,g(0))$-shyg exactly as required by
the theorem. 

So, let $U_1, \dots, U_{3k+c+1},S_1, \dots, S_{g(c)}$
be the considered $c$-strong shyg. 
Assume first that there is $u \in \bigcup_{i \in {3k+c+1}} U_i$
that touches $g(c-1)$ sets $S_j$. 
We assume w.l.o.g that $u \in U_{3k+c+1}$ and 
that the sets $S_j$ touching $u$ are precisely
$S_1, \dots, S_{g(c)-1}$. Since one intersection point 
between these sets is spent on $U_{3k+c+1}$
for any $j \neq j' \in [g(c)-1]$, 
$|S_j \cap S_{j'} \cap \bigcup_{i \in [3k+c]} U_i| \leq c-1$.
In other words, $U_1, \dots, U_{3k+c},S_1, \dots,S_{g(c)-1}$
is a $(c-1)$-strong shyg implying the theorem by the 
induction assumption. 

It remains to assume that each $u \in \bigcup_{i \in 3k+c+1} U_i$
touches at most $g(c-1)-1$ sets $S_j$. 
We are going to identify $I \subseteq [g(c)]$ of size 
$g(0)$  so that for each $j \neq j' \in I$, 
$S_i \cap S_j \cap \bigcup_{i \in [3k+1]} U_i=\emptyset$. 
This means that $U_1, \dots, U_{3k+1}$ along with $S_j$
for each $j \in I$ will form a stron $(3k+1,g(0))$-shyg as 
required by the theorem. 

We use the same elementary argument as if we wanted
to show that a graph with many vertices and a small max-degree
contains a large independent set. The initial set $CI$ of candidate
indices is $[g(c)]$ and initially $I=\emptyset$. We choose $j \in CI$ into $I$ and remove from 
$CI$ the $j$ and all the $j'$ such that $S_j \cap S_{j'} \cap \bigcup_{i \in [3k+1]} U_i \neq \emptyset$
By definition of $g(c)$, it is enough to show that, apart from $j$ itself,
we remove at most $(g(0)-1)(g(c-1)-2)$ other elements. 
Indeed, the size of $S^*_j=S_j \cap \bigcup_{i \in [3k+1]} U_i$ 
is at most $(3k+1)d=g(0)-1$ and each point of $S^*_i$, apart
from $S_j$ touches at most $g(c-1)-2$ elements simply by assumption. 
\end{proof}

\subsection{Constructing Gap Cover Approximators (\Cref{setapprox}) }

Our overall plan for the proof of \Cref{setapprox} is to show that we can produce the elements that make up the desired gap cover approximator, as a combination of four parts, each of which we can bound appropriately. The resulting algorithm is then primarily a matter of enumerating all combinations of elements from these parts. In the following we discuss the construction of these parts and why this  yields an \fpt algorithm.

The first part is what we will refer to as the set $\bep(U)$ of $p$-big edges w.r.t. $U$, which are those edges $e \in E(H[U])$ for which $|e| > pd$.
The intuition for the importance of big edges is simple, in a $(2,d)$-hypergraph, they are necessary to obtain low weight covers. In particular, any edge cover of $U$ with weight at most $p$ must contain all edges of $\bep(U)$: if $|e| > pd$ then the vertices in $e$ cannot be covered by less than $p+1$ other edges, since any other $e' \neq e$ will only intersect $e$ in at most $d$ vertices (see~\Cref{sec:covers} in the for in-depth discussion of $\bep(U)$ sets).
For the rest of this section we will simply say that edges are big to mean $p$-big. Similarly, we will refer to all edges that are not $p$-big as small.

The more challenging part is to determine the structure of those vertices that are not covered by the big hyperedges. To this end we first define the boundary $\bep^*$ of the big edges:
\[\bep^*(U) := \{e \setminus \bigcup (\bep(U) \setminus \{e\}) \mid e \in \bep(U) \}.\]
That is $\bep^*(U)$ contains those subedges of big edges that are unique to a single big edge.

With respect to this set we will be particularly interested in those members that together cover the intersection of a small edge with the vertices in the boundary.
We formalise this via the \emph{spanning set} $sp(e)$ for $e \in E(H[U]) \setminus \bep(U)$, which is the set $E' \subseteq \bep^*(U)$ such that $e \cap \bigcup E' = e \cap \bigcup \bep^*(U)$. We will refer to the cardinality of $sp(e)$ as the \emph{span} of $e$.
The final part we need is those vertices $U_p^0$ that are not part of any subedge in the boundary $\bep^*$, formally
$U_p^0 = U \setminus  \bigcup \bep^*(U)$.
In addition to the three parts described above, may need to add subsets of $U^0_p$ to construct the elements of the gap cover approximator. In particular, to add those vertices that are not part of any big edge. The key observation here is that, under the assumption that $\rho(U) \leq p$, there cannot be too many such vertices and specifically $|U^0_p| = O(p^2d)$.

Ultimately, what we prove is that for any $U' \subseteq U$ with $\rho(U') \leq k_0$, there is a set $X$ such that $X\supseteq U'$ and $\rho(X)\leq (3k+1+d)k_0$, where $X$ is the union of big edges, short edges  and $Y \subseteq U^0_p$.
The short edges are actually split in two cases, depending on their span. The construction of $X$ may require some number of short edges with span at most $3k+d$, as well some short edges with span greater than $3k+d$. The last set is the most challenging in terms of achieving an \fpt algorithm. All other sets can be bounded in terms of $p$ and $d$ (a small edge with a small span is covered by the union of its span leading to the stated approximation factor). Such a bound seems to not be achievable for the set of small edges with large span. To get around this issue, we show that if there are many small edges with large span, then this implies the existence of a large subedge hypergrid. In more concrete terms, if there is a set $E' \subseteq \bep^*(U)$ with $|E'| > 3k+d$ and $\mathit{sp}^{-1}(E') > \xi(k,d)$ (with $\xi$ as in the definition of $\mathbf{S}_{k,d}(H)$), then $\mathbf{S}_{k,d}(H) \neq \emptyset$. In consequence, we can detect the case for which we could not achieve an \fpt bound, and we know that in this case we can safely reject as $\ghw(H)$ is guaranteed to be greater than $k$.

\section{Discussion \& Future Work}
\label{sec:conclusion}

\textbf{Improvements to the approximation factor.}
Our main result represents a first step into the area of \fpt algorithms for \ghw. Our focus has been on developing the overarching framework, and we expect that with further refinement, better approximation factors are achievable and present a natural avenue for further research. 
The most immediate question is whether subcubic approximation can be achieved. Recall that the $(3k+d+1)(2k-1)k_0$ factor for \textsc{ApproxSep} comes from two sources, $(3k+d+1)k_0$ is a result of using $(3k+d+1)k_0,k_0$-gap cover approximators to find partial covers. The factor $(2k-1)$ is a result of combining the partial separators. It is unclear whether either of these factors can be avoided. 

\smallskip
\noindent
\textbf{Approximation of fractional hypertree width.}
A \emph{fractional edge cover} for hypergraph $H$ is a function $f$ that maps $E(H)$ to $[0,1]$ such that for every $v \in V(H)$ we have $\sum_{e\in I(v)} f(e) \geq 1$. The weight of a fractional edge cover is $\sum_{e\in E(H)} f(e)$ and $\rho^*(H)$ is the least weight of a fractional edge cover for $H$. Fractional hypertree width is defined just as \ghw, with the only difference that the width of a tree decomposition $(T,\B)$ is $\max_{u \in V(T)} \rho^*(\B(u))$. Clearly $\mathit{fhw}(H) \leq \ghw(H)$.
The notion is of similar importance as \ghw and typically when a problem is tractable under bounded \ghw, it is also tractable for bounded $\mathit{fhw}$. 

For $(c,d)$-hypergraphs it is known that there is a function $g$ such that $\rho^*(H) \leq k$ implies $|\rho(H)| \leq g(k,c,d)$~\cite{DBLP:conf/mfcs/GottlobLPR20,DBLP:journals/jacm/GottlobLPR21}. Through this bound, our main result directly implies the existence of an approximation algorithm for $\mathit{fhw}$. However, the currently known bounds are very high (in fact, it is unclear if the current bound is elementary). It is therefore a natural question for future work if our techniques can be adapted to an algorithm specifically designed to check $\mathit{fhw}$. Interestingly, it is easier to compute $\rho^*$ than $\rho$, as the latter is equivalent to \textsc{Set-Cover}, while the former is its natural linear programming relaxation. Nonetheless, enumeration of small fractional covers is challenging, and it is unclear whether an analogue of \Cref{setapprox} exists in the fractional setting.

\smallskip
\noindent
\textbf{Generalisation beyond $(2,d)$-hypergraphs.}
A final natural question for future work is the extension beyond $(2,d)$-hypergraphs. The natural next step here is an \fpt approximation algorithm for $(c,d)$-hypergraphs. While the step is natural on a conceptual level, on a combinatorial level the situation typically become much more complex for $c > 2$. The key problem being that nothing can be assumed about intersections of few edges. As such, techniques for $(2,d)$-hypergraphs do not easily generalise to $(c,d)$-hypergraphs and we expect this generalisation to be challenging. %

\newpage
\bibliography{refs}

\appendix

\section{On Subedge Hypergrids (\Cref{nogrid})}
\label{sec:subedgegrid}
\label{sec:shyg.ghw}

The goal of this section is to prove our main result on $(3k+d+1,\xi(k,d))$-shygs of $H$, namely~\Cref{nogrid}. Recall that 
we write $\mathbf{S}_{k,d}$ for the set of all such shygs of $H$.
Recall also that a set $U \subseteq V(H)$ is a \emph{subedge} (of $H$) if there is $e \in E(H)$ such that
$U \subseteq e$ and that two subedges $U_1,U_2$ are \emph{incompatible} if their union $U_1 \cup U_2$ is not a subedge. 

A key observation towards these results is that for large enough subedges, the witnessing edge is uniquely determined.
\begin{proposition}
  \label{subedge.uniq}
  If $U$ be a subedge of a $(2,d)$-hypergraph $H$ with $|U|>d$. Then there is a \emph{unique} $e \in E(H)$ such that $U \subseteq e$.
\end{proposition}
\begin{proof}
 By assumption that $U$ is a subedge, there is at least one $e \in E(H)$ with $U \subseteq e$. If there were another such $e' \in E(H)$, then $|e \cap e'| \geq  |U| > d$, a contradiction to $H$ being a $(2,d)$-hypergraph.
\end{proof}
For subedge $U$ where~\Cref{subedge.uniq} applies, we will write $e(U)$ for the unique \emph{witnessing edge} that contains $U$.

We will first focus on proving~\Cref{nogrid}, that any $(2,d)$-hypergraph $H$ for which $\mathbf{S}_{k,d}(H)\neq \emptyset$ has $\ghw$ greater than $k$. The technical challenge here lies with the possibility that there can be significant additional structure to the hypergraph on top of the shyg, e.g., large edges that connect many vertices in the shyg. It should be noted that that this is in strong contrast to the treewidth setting: for treewidth it is sufficient to identify a subgraph (or minor) with high treewidth to conclude that the graph itself has high treewidth. However, generalised hypertree width (and related hypergraph widths) is not hereditary (in terms of subhypegraphs), and little is known of how minors in hypergraphs relate to width (see~\cite{DBLP:conf/pods/Lanzinger22} for early results in this direction).
Our goal here will be to show that, in $(2,d)$-hypergraphs, large enough shygs, contradict the existence of the separators from~\Cref{balancedsep}, no matter what other edges are present in the hypergraph.

\begin{proposition}
  \label{propmutincomp}
  Let $U_1,\dots,U_q$ be subedges of cardinality greater than $d$ such that $e(U_1),\dots,e(U_q)$ are all distinct. Then $U_1,\dots,U_q$ are pairwise incompatible.
\end{proposition}
\begin{proof}
  Suppose there are $1\leq i \neq j\leq q$ and $e\in E(H)$ s.t. $U_i \cup U_j \subseteq e$. By~\Cref{subedge.uniq}, there is a unique $e(U_i \cup U_j)$ that contains $U_i \cup U_j$. By uniqueness of the witnessing edges of large subedges, $e(U_i \cup U_j) = e(U_i) = e(U_j)$, contradicting the assumption of distinct witnessing edges.
\end{proof}

\begin{lemma}
  \label{shyg.tech.1}
  Let $U_1,U_2, S_1,\dots,S_b$ be pairwise incompatible subedges of $(2,d)$-hypergraph $H$ such that:
  \begin{enumerate}
  \item $U_1 \cap U_2 = \emptyset$,
  \item for each $j\in[b]$, $S_j$ touches $U_1,U_2$, and
  \item for all distinct $i,j \in [b]$, $U_1 \cap S_i \cap S_j = \emptyset$ and $U_2 \cap S_i \cap S_j = \emptyset$. \label{shyg.tech.1.3}
  \end{enumerate}
  Let $e_1,\dots,e_r \in E(H)$, that are not witnessing edges for either $U_1$ or $U_2$, and such that $\bigcup_{i=1}^r e_i$ is a $(U_1,U_2)$-separator. Then $r \geq b/(2d)$.
\end{lemma}
\begin{proof}
  For $i \in [r]$, define $I_i$ as the set of all indexes $j$ such that $e_i$ has non-empty intersection with either $U_1 \cap S_j$ or $U_2 \cap S_j$. We first argue that $|I_i| \leq 2d$ for all $i \in [r]$. By~\Cref{shyg.tech.1.3}, any vertex $v \in e_i \cap U_1$, can be in at most one $S_j$. Since $e_i$ is not a witness of $U_1$ and $H$ has intersection width $d$, we have $|e_i \cap U_1| \leq d$. Thus, $e_i$ has non-empty intersection with $U_1 \cap S_j$, for at most $d$ subedges $S_j$. The same argument applies to $U_2$, demonstrating $|I_i|\leq 2d$.

  Now, if $r < q/(2d)$, then also $\left| \bigcup_{i=1}^r I_i \right| < q$. That means that there is some $j\in[q]$ that is not in any of the sets $I_i$. That is, there is a $S_j$ such that $(U_1 \cup U_2) \cap S_j$ does not touch any $e_i$. Therefore, there is a path from $U_1$ to $U_2$ (via $S_j)$ that does not touch $W$, contradicting the assumption that $W$ is a $(U_1,U_2)$-separator.
\end{proof}

We will first show the lower bound on $\ghw$ for a stronger version of shygs. Namely, we say that a $(a,b)$-shyg $U_1,\dots,U_a,S_1,\dots,S_b$ is \emph{strong}, if $S_j \cap S_{j'} \cap \bigcup_{i=1}^a U_i = \emptyset$ for all $j\neq j' \in[b]$. We will then show that existence of a $\akd$-shyg always implies the existence of a strong shyg of smaller dimension.

\begin{lemma}
  \label{shyg.cover.trick}
  Let $k\geq 1$, let $U$ be a subedge of $(2,d)$-hypergraph $H$ with $|U|> (3k+1)d$. Let $e_1,\dots,e_{3k} \in E(H)$ (not necessarily distinct) such that their union is a superset of $e(U)$. Then $e(U) = e_i$ for some $i \in [3k+1]$.
\end{lemma}
\begin{proof}
  Clearly $e(U)$ is unique since $|U|>d$. Suppose the statement is false, then $|e_i \cap e(U)| \leq d$ for all $i\in [3k+1]$. Hence, the union $\bigcup_{i=1}^{3k+1} e_i$ covers at most $3kd$ vertices of $e(U)$, which is not possible since $|e(U)|\geq |U|>3kd$.
\end{proof}

We are now ready to formally prove \Cref{strongnogrid} that was introduced in the main body. We recall the statement for convenience.
\RESTATEstrongnogrid*
\begin{proof}
  Let $a=3k+1, b=(3k+1)d+1$, and let $U_1,\dots,U_a,S_1,\dots,S_b$ be the assumed strong $(a,b)$-shyg.
  Clearly, for each $i\in[a]$,  $U_i$ contains at least $b$ vertices, and thus has a unique witnessing edge $e(U_i)$.

  Our plan is to show that no separator as in~\Cref{balancedsep} exists for the set $E'=\{e(U_1),\dots,e(U_a)\}$. To that end, suppose $\ghw(H)\leq k$, then~\Cref{balancedsep} applies and there is a weak partition $E'_0,E'_1,E'_2$ of $E'$ such that $E'_1,E'_2$ contain at most $2k$ edges, and there is a $(\bigcup E'_1, \bigcup E'_2)$-separator $W$ with $\bigcup E'_0 \subseteq W$ and $\rho(S) \leq k$.

  Assume, w.l.o.g., that $W$ is the union of at most $k$ hyperedges $E^* = \{e^*_1,\dots,e^*_\ell\}$.
  Let $X = E^* \cap E'$ and observe that $E'_0 \subseteq  X$: we have $\bigcup E'_0  \subseteq W = \bigcup_{i=1}^\ell e^*_i$ and since $\ell \leq k$, \Cref{shyg.cover.trick} applies and
  we have that every $e(U_i)\in E'_0$ will equal some $e^*_j$.

  We claim that $E'_1\setminus X$ and $E'_2\setminus X$ are both non-empty. If $E'_1 \subseteq X$, then $X \cup E'_2 = E'$ as we have already established that also $E'_0 \subseteq X$. Yet $|X| + |E'_2| \leq k + 2k < 3k+1$ and we see that this is not possible. The symmetric argument applies to the case where $E'_2 \subseteq X$.

  Let $U_i, U_j$ be subedges such that $e(U_i) \in E'_1 \setminus X$ and $e(U_j) \in E'_2 \setminus X$.
  Recall that the witnesses for $U_i,U_j$ are unique, and therefore the subedges are not witnessed by any of the $e^*_1,\dots,e^*_\ell$ (as they are all either in $X$ or not in $E'$). Together with the properties of a strong shyg, we see that~\Cref{shyg.tech.1} applies to $U_1,U_2,S_1,\dots,S_b$ and $e^*_1,\dots,e^*_\ell$. But by the lemma, $\ell \geq ((3k+1)d+1)/2d > k$ and we have a contradiction.
\end{proof}

\begin{lemma}
  \label{shyg.tech.2}
  Let $U_1,\dots,U_a,S_1,\dots,S_b$ be a $(a,b)$-shyg of $(2,d)$-hypergraph $H$.
  For $v \in \bigcup_{i=1}^a U_i$, let $I(v)$ be the set of all $j$ such that $v \in S_j$.
  Assume that there is an $r > 0$ such that  $(a\cdot d)^2 r < b$ and for each $v \in \bigcup_{i=1}^a U_i$ it holds that $|I(v)| \leq r$. Then $H$ has a strong $(a,\,ad+1)$-shyg.
\end{lemma}
\begin{proof}
  Let $U = \bigcup_{i=1}^a U_i$.
  For index $j \in [b]$, let $N(j) = \bigcup_{v\in S_j \cap U} I(v)$. Intuitively, this is the (closed) neighbourhood of index $j$, in the sense that it contains all indexes $j'$ of sets $S_{j'}$ that share a vertex with $S_j$.
  
  Now let $Q$ be the output of the following simple algorithm:
  \begin{enumerate}
  \item Set $Q = \emptyset, B = [b]$.
  \item Repeat until $B \neq \emptyset$: pick an $i \in B$, and update $Q = Q \cup\{i\}$ and $B = B\setminus N(i)$.
  \item Return $Q$.
  \end{enumerate}
  We will show that the restriction to $U$ of the sets selected in $Q$, will form the desired strong $(a,\,ad+1)$-shyg. Let $\mathbf{S}' = \{ S_j \mid j \in Q\}$. First, observe that $|\mathbf{S}'| \geq ad+1$. Since $|S_j \cap U_i| \leq d$ for all $i\in[a], j\in[b]$, and each $I(v)$ contains at most $r$ vertices by assumption, we have $|N(j)|\leq r\cdot a \cdot d$. Since $b > (a\cdot d)^2 r$, the algorithm will loop at least $ad+1$ times, adding a new element to $Q$ each time.

  By the algorithm, we have that for $j\neq j'\in Q$ and for all $i \in[a]$ that $U_i \cap S_j$ and $U_i \cap S_{j'}$ are disjoint: if some $v \in S_j \cap S_{j'} \cap U_i$, then $j',j \in I(v)$ and therefore also $j' \in N(j)$ and $j \in N(j)$. Clearly, this is impossible by construction. Thus, $U_1,\dots,U_a,\mathbf{S}'$ forms a strong $(a,\,ad+1)$-shyg.
\end{proof}

The preceding lemma now motivates the definition of function $\xi$, that we used to define the set $\mathbf{S}_{k,d}(H)$ consisting of all $(3k+d+1,\xi(k,d))$-shygs. For $n\geq 0$, we define the function recursively as  $\xi'(0,k,d)= (3k+1)d+1$ and $\xi'(n,k,d) = ((3k+1)d)^2 \xi'(n-1,k,d)+1$. The function $\xi$ used in the definition of $\mathbf{S}_{k,d}$ is then defined as $\xi(k,d) = \xi'(d,k,d)$. The next lemma will further clarify this choice.

\begin{theorem}
  \label{shyg.tech.thm}
  Let $n \geq 0$ and let $U_1,\dots,U_{3k+1+n}, S_1,\dots,S_b$ be a $(a,b)$-shyg of $(2,d)$-hypergraph $H$.
  Suppose that $b > \xi'(n, k, d)$ and for all $j\neq j' \in [q]$ we have $|S_j \cap S_{j'} \cap (U_1 \cup \cdots U_{3k+1+n})| \leq n$. Then $H$ has a strong $(3k+1, (3k+1)d+1)$-shyg.
\end{theorem}
\begin{proof}
  Proof is by induction on $n$. For $n=0$, we have immediately by assumption that $U_1,\dots,U_a, S_1,\dots,S_b$ is a strong $(3k+1,\xi'(0,k,d))$-shyg. Since $\xi'(0, k, d)=(3k+1)d+1$ we can apply~\Cref{strongnogrid} to see $\ghw(H) > k$.

  For $n > 0$, assume the statement holds for $n' = n-1$.
  Let $U = \bigcup_{i=1}^{3k+1+n} U_i$ and $I(v)$ be the set of $j$ such that $v \in S_j$ for all $v \in U$. We distinguish two cases. First, suppose there is a $v \in U$ with $|I(v)| > \xi'(n-1,k ,d)$. Assume, w.l.o.g., that $v \in U_{3k+1+n}$. Let $j\neq j' \in I(v)$ and observe that
  \[
    |S_{j} \cap S_{j'} \cap \bigcup_{i=1}^{3k+n} U_i| = |S_{j} \cap S_{j'} \cap \bigcup_{i=1}^{3k+n+1} U_i| - |S_{j} \cap S_{j'} \cap U_{3k+1+n}| \leq n - 1.
  \]
  The equality follows from the disjointness of the subedges $U_1,\dots,U_{3k+n+1}$. The inequality from $j,j' \in I(v)$, $v \in U_{3k+1+n}$ and the assumption that $|S_{j} \cap S_{j'} \cap \bigcup_{i=1}^{3k+n+1} U_i| \leq n$. We can then apply the induction assumption to $U_1,\dots,U_{3k+n}$, $\{S_j \mid j \in I(v) \}$.

  For the second case, for all $v\in U$ we have $|I(v)| \leq \xi'(n-1,k,d)$. In this case~\Cref{shyg.tech.2} applies (with $r = \xi'(n-1,k,d)$ and using only $U_1\dots,U_{3k+1}$) and $H$ has a strong $(3k+1, (3k+1)d+1)$-shyg.
\end{proof}

We are finally ready to prove~\Cref{nogrid}. We recall the statement for convenience.
\RESTATEnogrid*
\begin{proof}
  Recall, that $\mathbf{S}_{k,d}(H)$ is shorthand for the set of $(3k+1+d, \xi(k,d))$-shygs of $H$. By applying~\Cref{shyg.tech.thm} with $n=d$ we see that $H$ has a strong $(3k+1, (3k+1)d+1)$-shyg. Then by~\Cref{strongshyg}, $\ghw(H)>k$.
\end{proof}

\section{Efficiently Enumerating  Approximate Covers (\Cref{setapprox})}
\label{sec:covers}

In this section we expand on the details of computing gap cover approximators. We first recall the central definition of this section.
\RESTATEgapcovdef*

\medskip
In terms of our overall algorithm for approximating \ghw the following specific version of the problem of computing gap cover approximators will be relevant.
\begin{problem}{$f$-GCA}
  Input & Hypergraph $H$, $U\subseteq V(H)$, integers $k_0 \leq k \leq  \leq p$  and $\rho(U) \leq p$\\
  Parameter & $p + d$ \\
  Output & A $(f(k)k_0, k_0)$-gap cover approximator,
  \\ & or \textbf{Reject}, in which case it is guaranteed that $\ghw(H) > k$.
\end{problem}

In the rest of the section, we are going to prove 
that for $(2,d)$-hypergraphs, the problem \textsc{$(3k+d+1)$-GCA} is fixed-parameter tractable. As an important consequence this will also demonstrate that the size of the returned gap cover approximator has the same upper bound.

\begin{definition} \label{defbigedge}
  Let $H$ be a $(2,d)$-hypergraph and let $p$ an integer. The set $\bep(H)$ of \emph{$p$-big hyperedges} of $H$ is the set of all hyperedges of size greater than $p\cdot d$.
  We extend the definition to $U \subseteq V(H)$ by letting $\bep(U) = \bep(H[U])$.
\end{definition}

\begin{lemma} \label{bigbasic1}
  Let $H$ be a $(2,d)$-hypergraph, let $U \subseteq V(H)$, and let $p$ be an integer with $p \geq \rho(U)$. The following statements hold.
\begin{enumerate}
\item Let $\mu$ be an edge cover of $H[U]$ with weight at most $p$. Then for each $\bep(U) \subseteq \mu$.%
\item $|\bep(U)| \leq p$. 
\item For each $e \in \bep(U)$, let $e'= e \setminus \bigcup \left(\bep(U) \setminus \{e\} \right)$. 
Then $|e'| \geq d+1$. In particular, $e$ is the only edge of $H[U]$ that is a superset of $e'$. 
\end{enumerate}
\end{lemma}
\begin{proof}
  Let $e_1, \dots, e_p \in E(H[U])$,
  not necessarily all distinct, that cover $U$.
  Suppose there is $e \in \bep(U)$ that is not one
  of $e_1,\dots, e_p$. As for each $i \in [p]$,
  $|e_i \cap e| \leq d$, $|\bigcup_{i=1}^p e_i \cap e| \leq pd <|e|$,
  a contradiction proving the first statement. 

  Since each big edge must be among any weight $p$ edge cover,
  the number of big hyperedges cannot be greater than $p$.
  This proves the second statement. 

  For the last statement, let $e \in \bep(U)$. 
  Let $\bep(U) \setminus \{e\}=\{e'_1, \dots, e'_q\}$.
  Then, by the second statement, $q \leq p-1$. 
  observe that $e' = e \setminus \bigcup_{i=1}^q (e'_i \cap e)$ and since $e$ is $p$-big also $|e| > pd$, we have $|e'| > pd-(p-1)d=d$ as required. 
\end{proof}

\begin{definition} \label{defsubedges}
  Let $H$ be a $(2,d)$-hypergraph, let $U \subseteq V(H)$, and let $p$ an integer. Define
\begin{itemize}
\item $\bep^*(U)=\{e \setminus \bigcup (\bep(U) \setminus \{e\}) \mid e \in \bep(U)\}$,
\item $U^1_p=\bigcup \bep^*(U)$, and
\item $U^0_p=U \setminus U^1_p$.
\end{itemize}
\end{definition}

\begin{lemma} \label{bigbasic2}
    Let $H$ be a $(2,d)$-hypergraph, let $U \subseteq V(H)$, and let $p$ an integer with $p \geq \rho(U)$. Then
$|\bep^*(U)|=|\bep(U)|$ and, moreover, the
elements of $\bep^*(U)$ are disjoint mutually
incompatible subedges.
\end{lemma}
\begin{proof}
  Immediate from the last statement of \Cref{bigbasic1}
\end{proof}

\begin{definition}[Spanning set $sp_{U,p}$]
    Let $H$ be a $(2,d)$-hypergraph, let $U \subseteq V(H)$, and let $p$ an integer.
Let $e \in E(H[U]) \setminus \bep(U)$. 
The \emph{spanning set} $sp_{U,p}(e)$
 is the set $E' \subseteq \bep^*(U)$
such that $e \cap U^1_p=e \cap \bigcup E'$.
We may omit $p$ and $U$ in the subscript of $sp$ if they are clear from the context.
\end{definition}

\begin{lemma} \label{bigbasic3}
There is a $2^p \ \mathit{poly}(H)$ time algorithm
that computes $sp^{-1}_{p,U}$.
\end{lemma}
\begin{proof}
  Create a dictionary indexed by subsets of $\bep^*(U)$.
  The size of the dictionary is at most $2^p$ due to \Cref{bigbasic1}
  and \Cref{bigbasic2}. Initially each element of 
  the dictionary maps to the empty set. 
  Then loop over all the elements $e$ of $E(H[U]) \setminus \bep(U)$. 
  For each such $e$ identify the subset $E'$ of all elements of 
  $\bep^*(U)$ such that $e$ has a non-empty intersection with each
  $e' \in E'$. Add $e'$ to the set with index $E'$.
\end{proof}

\begin{lemma} \label{bigspan}  %
 Let $H$ be a $(2,d)$-hypergraph, let $U \subseteq V(H)$, let $p$ an integer with $p \geq \rho(U)$, and let $k \leq p$.
Suppose that there is 
a set $E' \subseteq \bep^*(U)$ 
such that $|E'| \geq 3k+1+d$
and
$sp^{-1}(E')>\xi(k,d)$. 
Then $ghw(H)>k$. 
\end{lemma}
\begin{proof}
  We may assume w.l.o.g. that 
  $|E'|=3k+d+1$ because otherwise we 
  can consider an arbitrary subset $E''$ of this
  size due to $sp^{-1}(E') \subseteq sp^{-1}(E'')$. 

  Let $U_1, \dots, U_{3k+d+1}$ be the elements of $E'$. 
  Let $S_1, \dots, S_q$ be the elements of $sp^{-1}(E')$. 

  First, we note $|U_i| \geq d+1$ for each $1 \leq i \leq 3k+d+1$
  by the last statement of Lemma \ref{bigbasic1}. 
  Next, we observe that $U_1, \dots, U_{3k+d+1}$ are
  mutually disjoint simply by definition of $\bep^*(U)$ and
  hence $|S_j| \geq d+1$ for each $1 \leq j \leq q$ as by definition,
  $S_j$ crosses more than $d+1$ mutually disjoint sets. 
  Thus $e_1, \dots, e_{3k+d+1}$ and $e(S_1), \dots, e(S_q)$ are all
  defined. 

  Due to the natural bijection $e \mapsto e \setminus \bigcup (\bep(U) \setminus \{e\})$ 
  between $\bep(U)$ and $\bep^*(U)$ according to Corollary \ref{bigbasic2},
  we conclude that $e(U_1), \dots, e(U_{3k+d+1})$ are all distinct. 
  Next $S_1, \dots, S_q$ are all distinct simply because they are distinct hyperedges.
  Finally each $e(U_i)$ is not the same as each $e(S_j)$ simply because the former is
  a big hyperedge while the latter is not.
  We thus conclude that $e(U_1), \dots, e(U_{3k+d+1}),e_1, \dots, e_q$ are all
  distinct. By \Cref{propmutincomp},
  $U_1, \dots, U_{3k+d+1},S_1, \dots, S_q$ are mutually incompatible. 

  Combined with the disjointness of $U_1, \dots, U_{3k+d+1}$ 
  and the assumption that each $S_j$ touches $U_1,\dots,U_{3k+d+1}$, we see that $H[U]$ has an $(3k+d+1, q)$-shyg with $q \geq \xi(k,d)$. By \Cref{nogrid}, then $ghw(H[U])>k$ and hence also $ghw(H)>k$. 
\end{proof}

\begin{lemma}  \label{smalluzero}
  Let $H$ be a $(2,d)$-hypergraph, let $U \subseteq V(H)$, and let $p$ an integer with $p \geq \rho(U)$. Then
$U^0_p$ is of size $O(p^2 \cdot d)$. 
\end{lemma}
\begin{proof}
  $U^0_p$ is the union of two sets: $\left( \bigcup \bep(U) \right) \setminus \left( \bigcup \bep^*(U) \right)$
  and $U \setminus \bigcup \bep(U)$ so it is enough to prove the upper bound 
  for each of them. 
  The former set is the union of intersections of all pairs of elements
  of $\bep(U)$: any vertex not in $\bigcup \bep^*(U)$ must be in at least two edges of $\bep(U)$. By Lemma \ref{bigbasic1}, $|\bep(U)| \leq p$ and the intersection
  size is at most $d$.
  
  For the second set, let $\{e_1, \dots, e_q\}$ be the smallest set of hyperedges covering
  $U \setminus \bigcup \bep(U)$. By assumption $q \leq p$. Since none of the edges in the cover are be $p$-big, they cover at most $q\cdot p \cdot d$ vertices in total. Hence, their union is of size at most $p^2d$. 
\end{proof}

\begin{algorithm}
    \DontPrintSemicolon
    \SetKwInOut{Input}{Input}\SetKwInOut{Output}{Output}
    \Input{$(2,d)$-hypergraph $H$, $U \subseteq V(H)$, $k_0 \leq k \leq  p$ and $\rho(U) \leq p$}
    \Output{A $((3k+d+1)k_0, k_0)$-gap cover approximator or \textbf{Reject}}
    \If{$\exists\, E' \subseteq \bep^*(U)$ such that $|E'|>3k+d$ and
      $sp^{-1}(E')>\xi(k,d)$}{ \label{effapp.if}
      \KwRet{\textbf{Reject}}\;
    }
    $X \gets \{E' \subseteq \bep^*(U) \mid 3k+d < |E'| \}$\;
    $LO \gets \bigcup_{E' \in X} sp^{-1}(E')$\;
    $RSet \gets \emptyset$ \;
    \For{$U' \subseteq U_p^0$ \text{ such that } $\rho(U') \leq k_0$}{ \label{Uloop}
      \For{$\mathit{big},\mathit{short},\mathit{long} \in \mathbb{N}$ such that $\mathit{big}+\mathit{short}+\mathit{long} \leq k_0$}{ \label{k0loop}
        \For{$B \subseteq \bep(U)$, $S \subseteq \bep^*(U)$ with $|B|=\mathit{big}$ and $|S| \leq (3k+d)\cdot \mathit{short}$}{ \label{bsloop}
          \For{$L \subseteq LO$ with $|L|=\mathit{long}$}{ \label{loloop}
            $RSet \gets RSet \cup \{U' \cup \bigcup (B \cup S \cup L)\}$\;
          }
        }
      }
    }
    \KwRet{$RSet$}\;
    \caption{The algorithm $\mathit{GapCoverApprox}(H,U,p,k,k_0)$}
    \label{effappbi}
\end{algorithm}

\begin{lemma} \label{effapprtheor}
  \Cref{effappbi} is a $f(p+d)\,\mathit{poly}(H)$ time algorithm for \textsc{ $(3k+d+1)$-GCA}.
\end{lemma}
\begin{proof}
Let us first verify the runtime of the algorithm.
The check at Line~\ref{effapp.if} can be done in $2^p \mathit{poly}(H)$ time
by \Cref{bigbasic3}. The set $X$ is straightforard to compute in time $O(2^p \cdot \mathit{poly}(H))$ since $|\bep^*(U)| \leq p$ by \Cref{bigbasic1} and \Cref{bigbasic2}.

If the algorithm does not reject after the first step,
then for each $E' \subseteq \bep^*(E)$
with $|E'|>3k+d$, it is guaranteed that $|sp^{-1}(E')| \leq \xi(k,d)$. 
It follows that 
\begin{equation} \label{eqapp1}
|LO| \leq 2^p \cdot \xi(k,d).
\end{equation}

The loop on Line~\ref{Uloop} performs at most $O(2^{p^2 d})$ iterations by \Cref{smalluzero}. Regarding the check $\rho(U')\leq k_0$, note that $H[U']$ has at most  $2^{|U'|}$ edges and we can check $\rho(U')$ in time $O(2^{{(p^2d)}^{k_0}}) = O(2^{p^3d})$ by naively checking all possible covers with $k_0$ edges.
We can overapproximate the number of iterations of the loop on Line~\ref{k0loop} generously by $O(p^3)$.
The loop on Line~\ref{bsloop} enumerates subsets of $\bep(U)$
and $\bep^*(U)$ so its numbers of iterations is $O(p^2)$ by \Cref{bigbasic1} and \Cref{bigbasic2}. 
Finally, the innermost loop on Line~\ref{loloop} is bounded by \eqref{eqapp1}. 

Let us now verify correctness of the output.
The correctness of the \textbf{Reject} case follows immediately from \Cref{bigspan}.
To verify correctness of $RSet$, i.e., that it indeed is a $((3k+d+1)k_0,k_0)$-gap cover approximator,
let us first show that for each $X \in RSet$, $\rho(X) \leq (3k+d+1)k_0$. 
By construction $X=U' \cup \bigcup (B \cup S \cup L)$. 
We have $\rho(U') \leq k_0$ and by the conditions on the respective loops, the sets $B$, $S$, and $L$ have cardinalities such that 
$|B \cup S \cup L| \leq big+short(3k+d)+long 
 \leq (3k+d)k_0$. 
Hence $\rho(X) \leq (3k+d)k_0+k_0=(3k+d+1)k_0$,
by adding the at most $k_0$ additional vertices from $U'$.

For the second property of gap cover approximators, let $W \subseteq U$ with $\rho(W) \leq k_0$.
We are going to demonstrate that $W$ is a subset of an 
element of $RSet$ and this will complete the theorem. 

Let $W^1=W \cap U^1_p$ and $W^0=W \cap U^0_p$. 
Let $\{e_1, \dots e_r\}$ be a $\rho$-stable set of hyperedges of $H[U]$ covering $W^1$ and note that 
$r \leq k_0$.
In the following we refer to hyperedges that are not $p$-big as \emph{$p$-small} (always w.r.t. $H[U]$).
The set $\{e_1, \dots, e_r\}$ can be weakly partitioned into sets
$Big$, $Short$, and $Long$ as follows: $Big$ is the set of the $p$-big hyperedges (w.r.t. $H[U]$), $Short$ is the set of all $p$-small hyperedges $e$ such that
$|sp(e)| \leq 3k+d$, and, hence $Long$ is the remaining set of $p$-small hyperedges
$e$ such that $|sp(e)| > 3k+d$. Note that because this is a (weak) partition, $|Big|+|Short|+|Long| \leq r \leq k_0$.
Let $SH^*=\bigcup_{e \in Short} sp(e)$.
It is not hard to see that 
$U^1_p \cap \bigcup Short \subseteq \bigcup SH^*$  
and hence 
$W \subseteq (W \setminus W^1) \cup \bigcup (Big \cup SH^* \cup Long)$

We are going to demonstrate that the set
$(W \setminus W^1) \cup \bigcup (Big \cup SH^* \cup Long)$
an element of $RSet$. 
Clearly $\rho(W \setminus W^1) \leq k_0$ and
hence $W \setminus W^1$ is selected in one of the iterations of the first loop. 
Let us fix this iteration and move to the second loop. 
It is also not hard to see that the loop on Line~\ref{k0loop} has an iteration
with $\mathit{big}=|Big|$, $\mathit{short}=|Short|$ and $\mathit{long}=|Long|$.

Recall that $Big$ consits only of $p$-big hyperedges (w.r.t. $H[U]$), i.e., $Big \subseteq \bep(U)$.
Furthermore, every term $sp(e)$ in the definition of $SH^*$ is a subset of $\bep^*(U)$ and hence $SH^* \subseteq \bep^*(U)$.
Also, $|SH^*| \leq |Short| \cdot (3k+d)$ and hence $Big, SH^*$ are chosen as values for $B$ and $S$ in some iteration of the loop on Line~\ref{bsloop}.
Finally, $LO$ in the algorithm consists precisely of the edges for which $|sp(e)| > 3k+d$ and clearly the set $Long$ is picked for $L$ on some iteration of the innermost loop. 
Having made the choices, the algorithm adds 
$(W \setminus W^1) \cup \bigcup (Big \cup SH^* \cup Long)$ to $RSet$. 
\end{proof}

\begin{corollary}
  Let $H$ be a $(2,d)$-hypergraph, let $U \subseteq V(H)$, and let $k_0\leq k \leq \rho(U)$.
  There exists a $((3k+d+1)k_0, k_0)$-gap cover approximator of size $f(k+d)\,\mathit{poly}(H)$.
\end{corollary}
\section{Proofs of Theorems \ref{thm:approxghw} and \ref{thm:compress} }

\RESTATEmain*
\begin{proof}
In the proof we refer to Algorithm \ref{alg:main}
for the pseudocode of an algorithm solving the 
$4\alpha(k,d)$-\textsc{ApproxGHW} problem.

We observe first that the algorithm runs \fpt
time. Clearly, it is enough to verify the \fpt
time for a single iteration of the loop of Line 6. 
The condition at Line 9 can be checked in \fpt
time by Theorem \ref{fptsetcov} and Line 10 takes
\fpt time by Theorem \ref{thm:compress} (subject
to the correctness as proved in the next paragraph) . 

For the correctness proof, assume first that the algorithm
returns $(T_n,{\bf B}_n)$ at Line 14. 
We claim that for each $1<i \leq n$, $(T_i,{\bf B}_i)$ formed at Line 13
is a TD of $H_i$ of \fpt at most $4\alpha(k,d)$. 
For $i=2$, this is immediate by description. 
For $i>2$, we only consider the case where the condition of Line 9
is true for otherwise the statement is immediate. 
By adding one vertex to each bag, the \fpt of the resulting 
TD can grow by at most $1$.
Therefore, if the condition is true that the \fpt of $(T_i,{\bf B}_i)$ formed
at Line 7 and 8 is $4\alpha(k,d)+1$. 
This shows that the application of $\mathit{Compress}$ satisfies 
the premises on its input as specified in Algorithm \ref{alg:compress}. 
As Line 14 of Algorithm \ref{alg:main} is reached, the considered
run of $\mathit{Compress}$ does not reject and hence $(T_i,{\bf B}_i)$ formed
at Line 1 is of \fpt at most $4\alpha(k,d)$ by Theorem \ref{thm:compress}. 

It remains to assume that Algorithm \ref{alg:main} rejects at Line 12 of some iteration
$i$. The rejection can only be caused by the \textbf{Reject} returned by $\mathit{Compress}$
applied at the same iteration. With the same inductive reasoning as in the previous 
paragraph, we conclude that the input of $\mathit{Compress}$ is valid and hence
$ghw(H)>k$ follows from Theorem \ref{thm:compress}. 
\end{proof} 

We are now turning of the proof of Theorem \ref{thm:compress}.
\RESTATEfptcompress*
\begin{proof}
In the proof, we refer to Algorithm \ref{alg:compress} for the
pseudocode of the $\mathit{Compress}$ procedure. 
Whenever an input tuple of $\mathit{Compress}$ appears in a claim statement,
we assume that the tuple satisfies the premises of the input as specified
in Algorithm \ref{alg:compress}.

\begin{claim} \label{algo2clm1}
Assume that the algorithm does not reject at Line 4 and let $X$
be the output of $AppSep$ at Line 6. 
Then $H \setminus X$ has at least two connected components. 
\end{claim}
\begin{claimproof}
Let $E_1,E_2,E_0$ be weak partition of $E_W$ corresponding
to the considered application of $AppSep$. We observe that both $E_1$ and $E_2$ are nonempty. 
Indeed, suppose that $E_2=\emptyset$. Then $E_W=E_1 \cup E_0$.
On the one hand, $|E_W|=3\alpha(k,d)+1$, on the other hand, $|E_1| \leq 2\alpha(k,d)$, $|E_0| \leq k$, 
hence $|E_W| \leq 3\alpha(k,d)$, a contradiction.

Theorem \ref{fptsep} guarantees that $\rho(X) \leq \alpha(k,d)$
and that $\bigcup E_1 \setminus X$ and $\bigcup E_2 \setminus X$
are in different connected components. It remains to verify that
both these sets are non-empty. We apply argumentation similar to the
previous paragraph. Suppose that $\bigcup E_2 \setminus X=\emptyset$.
As $\bigcup E_0 \subseteq X$, it follows that $\bigcup E_W \subseteq \bigcup E_1 \cup X$.
However, the edge cover number of the left-hand set is $3\alpha(k,d)+1$, while the edge 
cover number of the right-hand set is at most $3\alpha(k,d)$, a contradiction.
\end{claimproof}

\begin{claim} \label{algo2clm2}
Suppose that $\mathit{Compress}(H,k,(T,{\bf B}),W)$ applies
itself recursively to a tuple $(H',k,(T,{\bf B}),W')$.
Then $|V(H')|<|V(H)|$, $\rho(W') \leq 3\alpha(k,d)$, 
and $\rho(H') \geq 4\alpha(k,d)+1$. 
\end{claim}
\begin{claimproof}
If $\mathit{Compress}(H,k,(T,{\bf B}),W)$ applies
itself recursively, it does not reject at Line 4.
Let $X$ be the set computed at Line 6.
Then $H'=H[V' \cup X]$ where $V'$ is a connected
component of $H \setminus X$. 
By Theorem \ref{fptsep},
$\rho(\bigcup E_W \cap V') \leq 2\alpha(k,d)$
(as the set can be covered by either $E_1$ or $E_2$)
and $\rho(X) \leq \alpha(k,d)$. 
Hence,  according to Line 11, $\rho(W') \leq 3\alpha(k,d)$. 
Finally,  $\rho(H')$ being at least $4\alpha(k,d)+1$ is immediate from Line 7. 
\end{claimproof}

Claim \ref{algo2clm2} allows us to use induction
assumption on $|V(H)|$ for recursive calls made
by compress: it shows that the number of vertices
decreases and that the input of the recursive 
application satisfies the input constraints of Algorithm \ref{alg:compress}.
In the rest of the proof we use the induction assumption
without explicit reference to  Claim \ref{algo2clm2}. 

\begin{claim} \label{algo2clm3}
Let $k,a_1, \dots, a_q$ be integers such that $k \geq 1$, $q \geq 2$, and
$a_i \geq 3k$ for each $1 \leq i \leq q$.  
Then $\sum_{i=1}^q (a_i+k)^2<(\sum_{i=1}^q a_i)^2$. 
\end{claim}

\begin{claimproof}
Let $r_i=a_i/k$ for $1 \leq i \leq q$. 
Assume first that $q=2$. 
Then we need to show that $(r_1k+k)^2+(r_2k+k)^2 < (r_1k+r_2k)^2$. 
On the left-hand side, we obtain
$r_1^2k^2+r_2^2k^2+2r_1k^2+2r_2k^2+2k^2$.
On the right-hand side we obtain 
$r_1^2k^2+r_2^2k^2+2r_1r_2k^2$.
After removal of identical items and dividing both parts
by $2k^2$, it turns out that we need to show that
$r_1+r_2+1<r_1r_2$. But it is not hard to see that this is indeed so
if both $r_1$ and $r_2$ are at least $3$.

Assume now that $q>2$. 
Then by the previous paragraph,
$\sum_{i=1}^{q-2} (a_i+k)^2+(a_{q-1}+k)^2+(a_q+k)^2<\sum_{i=1}^{q-2} (a_i+k)^2+(a_{q-1}+a_q)^2<
\sum_{i=1}^{q-2} (a_i+k)^2+(a_{q-1}+a_q+k)^2<(\sum_{i=1}^q a_i)^2$, the last inequality follows
from the induction assumption.
\end{claimproof}

\begin{claim} \label{algo2clm4}
$\mathit{Compress}(H,k,(T,{\bf B}),W)$ makes at most $\rho(H)^2 \cdot |V(H)|$ recursive
applications. 
\end{claim}

\begin{claimproof}
The recursive calls of $\mathit{Compress}$ can be naturally organized into a rooted \emph{recursion tree}.
The initial application corresponds to the root of the tree. The children of the given 
application correspond to the recursive calls made directly within this application.
The applications that do not invoke any new recursive applications correspond to the 
leaves of the tree. It is immediate from Claim \ref{algo2clm2} that the height of the tree is
at most $|V(H)|$. Thus, it remains to verify that the number of leaves is at most $\rho(H)^2$. 

We apply induction on $|V(H)|$. If $|V(H)| \leq 4\alpha(k,d)+1$ then the function does not apply itself
recursively. Indeed, by Claim \ref{algo2clm2}, the number of vertices in any purported recursive
application would be at most $4\alpha(k,d)$ and hence the same applies on the edge cover number of
the graph. But this contradicts the last statement of Claim \ref{algo2clm2}. 
Hence, in this case there is only one leaf. 
If $\mathit{Compress}(H,k,(T,{\bf B}),W)$ makes only one recursive application then the
subtree corresponding to this application has the same leaves as the tree
corresponding to $\mathit{Compress}(H,k,(T,{\bf B}),W)$. Hence, the statement holds
by the induction assumption. 

Assume now that $\mathit{Compress}(H,k,(T,{\bf B}),W)$ makes $q \geq 2$ recursive
applications. This means that \textbf{Reject} is not returned at Line 4.
Let $X$ be the set computed at Line 6 and let $U_1, \dots U_q$ be the sets
computed at Line 7. By construction and the induction assumption,
the root of the recursion tree has $q$ children and the respective numbers
of leaves in the subtrees rooted by these children are at most
$\rho(V_i \cup X)^2$ for $1 \leq x \leq q$. We need to demonstrate that 
$\sum_{i=1}^q  \rho(U_i \cup X)^2 \leq \rho(H)^2$. 
By Line 7, for each $1 \leq i \leq q$, $\rho(U_i \cup X)>4\alpha(k,d)$
and, by Theorem \ref{fptsep}, $\rho(X) \leq \alpha(k,d)$. 
It follows that $\rho(U_i)>3\alpha(k,d)$. 
Hence, 
$\sum_{i=1}^q \rho(U_i \cup X)^2  \leq 
\sum_{i=1}^q (\rho(U_i)+\alpha(k))^2 \leq
(\sum_{i=1}^q \rho(U_i) )^2=\rho(\bigcup_{i=1}^q U_i))^2 \leq \rho(H)^2$,
the second inequality follows from Claim \ref{algo2clm3},
for the first equality observe that distinct $U_i,U_j$
do not intersect with the same hyperedge by the properties
of tree decompositions.
\end{claimproof}

\begin{claim} \label{algo2clm5}
$\mathit{Compress}(H,k,(T,{\bf B}),W)$  takes \fpt time
parameterized by $k$ and $d$. 
\end{claim}

\begin{claimproof}
In light of Claim \ref{algo2clm4}, we only need to prove
this claim for a single recursive application of $\mathit{Compress}$. 
To compute Line 1, gradually add vertices to $W$ 
until it becomes of edge cover number $3\alpha(k,d)+1$.
The fixed-parameter tractabilty of such a procedure follows
from Proposition \ref{fptsetcov}.
The fixed-parameter tractability of the remaining part of the
procedure follows from Proposition \ref{fptsep} and 
a straightforward analysis of pseudocode.
\end{claimproof}

\begin{claim} \label{algo2clm6}
If $\mathit{Compress}(H,k,(T,{\bf B}),W)$ returns \textbf{Reject}
then $ghw(H)>k$. 
\end{claim}

\begin{claimproof}
We are going to verify the claim for the \textbf{Reject}
returned at Line 4. Once done, the validity of \textbf{Reject}
at Line 15 will be straightforward: by the induction assumption
the \fpt of an induced subgraph of $H$ is greater 
than $k$. Hence, this is true for $H$ itself. 

By Theorem \ref{fptsep}, the return of \textbf{Reject} at Line 4
means that either $AppSep$ itself detects that $ghw(H)>k$ 
or that a balanced separator, exhaustively searched in Line 2
does not exist. 
In this case, $ghw(H)>k$
follows from Theorem \ref{balancedsep}. For the latter case, 
note that, due to the $\rho$-stability of $E'$, $E'_0$ must be
of size at most $k$ so that $\bigcup E'_0$ can be a subset of 
a set with the edge cover number at most $k$. Therefore, this 
assumption about $E'_0$ does not miss 'good' partitions during 
the search in Line 2. 
\end{claimproof}

\begin{claim} \label{algo2clm7}
If $\mathit{Compress}(H,k,(T,{\bf B}),W)$ does not 
return \textbf{Reject} then it returns
$(T^*,{\bf B^*})$, a TD of $H$ of \fpt
at most $4\alpha(k,d)$ with $W$ being a subset of some
bag of $(T^*,{\bf B^*})$. 
\end{claim}

\begin{claimproof}
First, we observe that $T^*$ is a tree by induction on $|V(H)|$.
If $\mathit{Compress}$ does not run recursively then $T^*$ is a star by construction. 
Otherwise, the node $r$ created at Line 8 is connected to exactly one node
of each tree (by the induction assumption) created by the recursive calls.
Clearly, the resulting graph is a tree. 

By Theorem \ref{fptsep}, $\rho(X) \leq \alpha(k,d)$ and
$\rho(W) \leq 3\alpha(k,d)$ by assumptions about the input.
Hence, considering Lines 8 and 22 and applying induction,
we observe that ${\bf B^*}: V(T^*) \rightarrow 2^H$ 
with the edge cover number of each bag at most $4\alpha(k,d)$ 
and $W \subseteq {\bf B^*}(r)$. Moreover, the induction application
of the last statement implies existence of the nodes $t_i$ as specified
in Line 18. 

It remains to verify that $(T^*,{\bf B^*})$ satisfies the containment
and connectedness properties. 
For the containment, we note that each $e \in E(H)$ is a hyperedge
of some $H[V' \cup X]$ where $V'$ is a connected component of $H \setminus X$. 
If $V'$ is a connected component considered in the loop of Line 20 then
$e$ is a subset of the corresponding bag by construction. 
Otherwise, $\mathit{Compress}$ applies itself recursively to a tuple where $H[V' \cup X]$
is the first parameter and hence the containment for $e$ follows from the induction
assumption. 

For the connectedness, assume first that
$\mathit{Compress}(H,k,(T,{\bf B}),W)$ does not apply itself recursively.
Then $T^*$ is a star and the only vertices that may appear in more than 
one bag are those of $W \cup X$, belonging to the bag of the centre of the star. 

Otherwise, $V(T)$ can be partitioned into the star as in the previous paragraph
and the trees $T'_i$ created in Line 17. As we have already observed the connectednes
of ${\bf B^*}$ restricted to the star is preserved and the connectedness of the 
restriction to each $T'_i$ holds by the induction assumption. 
It only remains to verify total connectedness for a vertex $u$ appearing
in more than one part of $T^*$ as specified. 
By construction, this vertex $u$ must be a subset of $W \cup X$. 
Then $u \in {\bf B^*}(r)$ and also must be an element of the bag of each $t_i$
whenever $u \in U_i \cup X$ \footnote{This is why the requirement of
$W$ being a subset of some bag is needed.}. It follows that the paths of bags containing $u$
in the individual $T'_i$ are all connected to $r$ (and if $u$ is in any of the bags
created during the loop of Line 20, the corresponding nodes are also connected
to $r$). Thus, the total connectivity is established. 
\end{claimproof}

The theorem follows from Claims \ref{algo2clm5}, \ref{algo2clm6}, and 
\ref{algo2clm7}.
\end{proof}
\section{Proof of Theorem \ref{fptsep}}
We restate the theorem here for the sake of convenience.
\RESTATEfptsep*

This section consists of two subsections.
In the first one we prove auxiliary claims.
The actual proof of Theorem \ref{fptsep}
is located in the second subsection.

We refer to Algorithm \ref{alg:sepgen}
for the pseudocode of $\mathit{AppSep}$ and 
to Algorithm \ref{alg:smallsep} for the pseudocode 
of $\mathit{SmallSep}$. 

\subsection{Auxiliary statements}
\begin{lemma} \label{numrecur}
For integers $n \geq 1, k \geq 1$ and some growing 
function $g$ with $g(1) \geq 1$, let $T(n,k)$ be defined 
as follows 
\begin{enumerate}
\item For $n \leq 2$, $T(n,k)=1$. 
\item Otherwise, for $k=1$, $T(n,k) \leq 2T(3n/4,k)$
(note that when we use $3n/4$ we actually mean $\lfloor 3n/4 \rfloor$).
\item Otherwise, $T(n,k) \leq 2T(3n/4,k)+g(k) \cdot T(3n/4,k-1)$.
\end{enumerate}

Then $T(n,k) \leq f(k) \cdot n^3$ where 
$f(1)=1$ and for $k>1$, $f(k)=2.7^{k-1} \cdot \prod_{i=2}^k g(i)$. 
\end{lemma}
\begin{proof}
Assume first that $k=1$. 
Then the statement holds for $n \leq 2$.
Otherwise, by the induction assumption, 
$T(n,1) \leq 2 \cdot (3n/4)^3=54/64n^3<n^3$ as required. 

Assume now that $k>1$.
If $n \leq 2$ then the statement is immediate by the properties 
of the function $g$.
Otherwise, assume by induction that the statement
holds for all $(n_0,k_0)$ with $n_0<n$, $k_0 \leq k$. 
Observe that for $k>1$,

\begin{equation} \label{runeq}
f(k)=54/64 f(k)+27/64 \cdot g(k) \cdot f(k-1)
\end{equation}

Indeed, this is the same as to say 
that
$10/64f(k)=27/64 \cdot g(k) \cdot f(k-1)$ or that
$f(k)=2.7 \cdot g(k) \cdot f(k-1)$.
It is easy to observe, by induction on $k$
that the last formulation for $f(k)$ is equivalent to
the one provided in the statement of the lemma. 

Now, by the induction assumption, 
$T(n,k) \leq 2f(k) (3n/4)^3+g(k) \cdot f(k-1) \cdot (3n/4)^3=
  54/64 f(k) n^3+27/64 \cdot g(k) \cdot f(k-1) \cdot n^3=
  (54/64\cdot f(k)+27/64 \cdot g(k) \cdot f(k-1)) n^3=f(k) n^3$,
the last equality follows from \eqref{runeq}.
\end{proof}

\RESTATEcomposedsep*
\begin{proof}
For the sake of brevity, let us denote $W_1 \cup W_2 \cup W$ by 
$W^*$. 
\begin{claim} \label{clm222}
$W^*$ is a $(C_1,C_2)$-separator of $H$. 
\end{claim}
\begin{claimproof}
Assume the opposite and let $P$ be the shortest
path between $C_1$ and $C_2$ in $H \setminus W^*$. 
Then the first vertex of $P$ belongs to $C_1$,
the last to $C_2$ and no vertex in the middle belongs
to $Y$ (because each vertex of $Y \setminus W$ is either of $C_1$ 
or of $C_2$ and hence a shorter path can be obtained). 
By definition of $Y$ as an $(V_1,V_2)$-separator, this means
that the rest of the vertices are either all in $V_1$ or all 
in $V_2$. In the former case, $P$ is a $C_1,C_2$ path 
of $H_1 \setminus (W_1 \cup W_2)$, in the latter case,
$P$ is a $C_1,C_2$-path of $H_2 \setminus (W_1 \cup W_2)$. 
In both cases, we get a contradiction to the definition of
either $W_1$ or $W_2$. 
\end{claimproof}

Assume that the lemma is not correct and let
$P$ be an $A,B$-path in $H \setminus W^*$. 
If $P$ lies completely in $H_1$ then $P$ is an
$A_1,B_1$-path in $H_1 \setminus W_1$, a contradiction. 
If $P$ lies completely in $H_2$ then $P$ is an
$A_2,B_2$-path of $H_2 \setminus W_2$, again a contradiction.
It follows that $P$ must include vertices of $Y$.

Let $a_1$ and $a_2$ be the first and the last vertices of $Y$
in $P$ (that might be the same vertex). 
Consider $P_1$, the prefix of $P$ ending at $a_1$. 
As $P_1$ does not contain vertices of $Y$ but $a_1$,
$P_1$ either wholly lies in $H_1$ or wholly lies in $H_2$. 
In the former case, $P_1$ is an $A_1,Y$-path in $H_1 \setminus W_1$.
By definition of $W_1$, $a_1 \in C_1$. 
In the latter case, $P_1$ is an $A_2,Y$-path and again ends up in $C_1$
by definition of $W_2$. 
Consider $P_2$, the suffix of $P$ starting at $a_2$.
Arguing symmetrically and using $B$ instead of $A$, we observe
that $a_2 \in C_2$. Thus, $P$ contains a subpath which is a $C_1,C_2$-path
contradicting Claim \ref{clm222}. 
\end{proof}

In order to proceed we need the definition of $T_{t,Y}$
and the related notions provided in Section \ref{sec:pseudo}

\RESTATEvalidtw*
\begin{proof}
The theorem consists of four statements. 
For simplicity of referencing, we enumerate 
them in the list below. 
\begin{enumerate}
\item $(T_{t,Y},{\bf B}_{t,Y})$ is a \textsc{td} of $H_{t,Y}$.
\item $(T^+_{t,Y},{\bf B}^+_{t,Y})$ is a \textsc{td} of $H$. 
\item $(T,{\bf B}^{-W})$ is a \textsc{td} of $H \setminus W$.
\item All the vertices of $V(T_{t,Y}) \setminus X_{t,Y}$
are leaves of $T_{t,Y}$.
\end{enumerate}

For the first statement, 
note that, due to the connectedness
of $T_{t,Y}$, for any two vertices of $T_{t,Y}$, the path in $T$
between them lies in $T_{t,Y}$.
Hence, the connectedness of $(T_{t,Y},{\bf B}_{t,Y})$ is immediate

For the containment, assume that $e \in E(H_{t,Y})$ 
is not a subset of any ${\bf B}_{t,Y}(t')$ for $t' \in V(T_{t,Y})$. 
By the containment property of $(T,{\bf B})$, there is  
$t_0 \in V(T)$ such that $e \subseteq {\bf B}(t_0)$.

Let $u \in e \setminus {\bf B}_{t,Y}(t)$, 
Then there is, $t_1 \in V(T_{t,Y})$
such that $u \in {\bf B}_{t,Y}(t_1)={\bf B}(t_1)$.
The path in $T$ between $t_1$ and $t_0$ goes through
$t$ in contradiction to the connectedness of $(T,{\bf B})$.
This proves the first statement. 

For the second statement, we note that 
each bag of $(T,{\bf B})$ is a subset of a bag of
$(T^+_{t,Y},B^+_{t,Y})$, hence the containment 
property is immediate.

For the connectedness property, 
let $t_1, t_2 \in V(T_{t,Y}^+)$ and suppose that there is
$u \in V(H)$ such that $u \in {\bf B}^+_{t,Y}(t_1) \cap
{\bf B}^+_{t,Y}(t_2)$. 
We need to show that $u$ belongs to all of the bags 
of the path between $t_1$ and $t_2$. We assume that the path
contains intermediate vertces for otherwise the 
statement is immediate.  
If both $t_1$ and $t_2$ are the nodes of $T_{t,Y}$
then we apply the same argument as in the 
connectedness proof for the first statement.
Otherwise, say $t_2$ is the vertex added to $T_{t,Y}$ 
to obtain $T^+_{t,Y}$. 
This means that there is $t_3 \in V(T) \setminus V(T_{t,Y})$ 
such that $u \in {\bf B}(t_3)$. 
Let $P$ be the path between $t_1$ and $t_3$ in $T$.
Then the intermediate vertices $t_0$ of the path between $t_1$ and 
$t_2$ in $T^+_{u,Y}$ form a subpath of $P$.
As ${\bf B}^+_{t,Y}(t_0)={\bf B}(t_0)$, $u \in {\bf B}^+_{t,Y}(t_0)$
due to the connectedness of $(T,{\bf B})$. This 
completes the second statement. 

For the third statement, we simply note that
the \textsc{td} properties of $(T,{\bf B}^{-W})$ 
straightforwardly follow from the \textsc{td}
properties of $(T,{\bf B})$. 

For the last statement, we observe that
$V(T_{t,Y}) \setminus X_{t,Y} =
 V(T_{t,Y}) \setminus X \subseteq V(T) \setminus X$. 
Thus any $t_0 \in V(T_{t,Y}) \setminus X_{t,Y}$ and, clearly,
its degree does not increase in $T_{t,Y}$. 
\end{proof}

\begin{lemma} \label{treesplit1}
Let $T$ be a tree of at least three vertices. 
Then there is a non-leaf vertex $x \in V(T)$ 
and a partition $Y_1,Y_2$ of $N(x)$ 
so that for each $i \in \{1,2\}$, $|V(T_{x,Y_i})| \leq 3/4 \cdot |V(T)|$.
Moreover, $x,Y_1,Y_2$ can be efficiently computed. 
\end{lemma}

\begin{proof}
\Cref{alg:subtree} provided below returns $(x,Y)$ where
$x \in V(T)$ and $Y \subseteq N(x)$. We are going to show that
$Y,N(x) \setminus Y$ is a desired partition of $N(x)$ as 
specified in the lemma. 

Throughout this proof, we make the following
notational convention. When we use $N(.)$ (without
a subscript), we mean $N_T(.)$. 
Also, $T_{x,-y}$ means $T_{x,N(x) \setminus \{y\}}$.  

\begin{algorithm}
    \DontPrintSemicolon
    \SetKwInOut{Input}{Input}\SetKwInOut{Output}{Output}
    \Input{a tree $T$ with at least three vertices}
    $i \gets 0$ \;
    Let $x_0$ be an arbitrary element of $V(T)$ \;
    $T^0 \gets T$ \;

  \While{there is $x \in N_{T_i}(x_i)$ with $|V(T_{x,-x_i})|>2/3|V(T)|$}{
     $i \gets i+1$\;
     $x_i \gets x$ \;
     $T^i \gets T_{x_i,-x_{i-1}}$ \; 
    }
  \If{there is $x \in N_{T_i}(x_i)$ with $V(T_{x,-x_i})>1/2|V(T)|$}{
        \KwRet{$(x,N(x) \setminus \{x_i\}$}\;
      }
  Let $y_1, \dots, y_q$ be the elements of $N(x_i)$ 
  listed in the non-increasing order of $|V(T_{y_j,-x_i})|$ \;
  Let $r \geq 1$ be the smallest number such that 
  $\sigma_{j=1}^r |V(T_{y_j,-x_i})|>1/3(|V(T)|-1)$ \;  
  \KwRet{$(x_i,\{y_1, \dots, y_r\})$}\;
   \caption{Computation of $(t,Y)$ as specified in \Cref{treesplit1}.}
    \label{alg:subtree}
  \end{algorithm}

\begin{claim} \label{clmtree1}
Let $x \in V(T)$, $Y \subseteq N(x)$ and $y \in Y$. 
Then $T_{y,-x}$ is a proper subtree of $T_{x,Y}$
\end{claim}
\begin{claimproof}
Just observe that any path of $T$ starting at 
$y$ and not including $x$ is also a path in 
$T_{x,Y}$.
\end{claimproof}

We conclude from Claim \ref{clmtree1} that 
for each $i$ that is not the largest value of the counter of
the main loop, $T^{i+1}$ is a proper subset of $T^i$.
This means that the loop at line $4$ stops.
Let $a$ be the last $i$ considered by the loop. 

\begin{claim}  \label{clmtree2}
Suppose that $a>0$.
Then $T_{x_{a-1},-x_a}$ has less
than $1/3|V(T)|$ vertices.  
\end{claim}
\begin{claimproof}
    Indeed, it is not hard to notice that 
$V(T^a) \cup V(T_{x_{a-1}, -x_a})=V(T)$,
and that the union is disjoint.
According to the condition of the loop
$V(T^a)$ includes more than two third of the vertices 
of $T$. 
\end{claimproof}

Let us say that $y$ is a \emph{big neighbour} of $x_a$
if $|V(T_{y,-x_a})| >1/2|V(T)|$. 
Let us call the condition of line $8$ of the algorithm
\emph{the big neighbour condition};

First let us verify that the lemma holds 
when $T$ has three vertices.
In this case $T$ is just a path $(x,y,z)$
According to the description $x_0=y$ and $a=0$.
The big neighbour condition is false.
The algorithm returns $(y,Y)$ where $Y$ contains precisely
one neighbour of $y$, say $x$.
Clearly, both $T_{y,\{x\}}$ and $T_{y,\{z\}}$ contain $2$ vertices
which is $2/3 \cdot V(T)$. So, in the rest of the proof we assume
that $|V(T)>3$. 

Assume that the big neigbhour condition is satisfied and let
$y$ be a big neighbour of $x_a$ (in fact, it is 'the neigbhour' but this 
is not relevant to the present discussion).
As the condition of the main loop is not satisfied, we conclude that 
$|V(T_{y,-x_a})| \leq 2/3|V(T)|$. 
For the other tree in question,
as $|V(T_{y,-x_a})|>1/2|V(T)|$, we conclude that 
$|V(T_{y,x_a})|<1/2|V(T)|+1$.
The maximum ratio of the right-hand  side and $|V(T)|$ is reached when 
$|V(T)|=5$ and it is $3/5$. So, we conclude that
$|V(T_{y,x_a})| \leq 3/5|V(T)|$. 
Thus the lemma holds in case of the big neighbour condition. 

It remains to assume that 
the big neighbour condition 
does not hold.
We observe that, in this case, $|V(T_{y_i,-x_a})| \leq 1/2|V(T)|$
for each $1 \leq i \leq q$. For those $y_i$ that are vertices of
$T^a$, this follows from the failure of both the loop and the big 
neighbour conditions. For $x_{a-1}$ (in case $a>0$), this follows from
Claim \ref{clmtree2}. 

Assume first that $r=1$. 
Due to the failure of the big neighbour condition, this means 
that $|V(T_{x_a,y_1})| \leq 1/2|V(T)|+1$. 
The maximum ratio of the right-hand side and $|V(T)|$ 
is reached when $|V(T)|=4$ and it is $3/4$.
So, we conclude that 
$|V(T_{x_a,y_1})| \leq 3/4|V(T)|$. 

Assume now that $r>1$.
That is $\sum_{i=1}^{r-1} |V(T_{y_i,-x_a})| \leq 1/3(|V(T)|-1)$.
Due to the ordering of $y_i$, this is an upper bound on
every individual $|V(T_{y_i,-x_a})|$. 
We conclude that
$|V(T_{x_a, \{y_1, \dots, y_r\}})| \leq 2/3(|V(T)|-1)+1$.
The ratio of the right hand part to $|V(T)|$ reaches maximum 
when $|V(T)|=4$ and it is $3/4$. We conclude that 
$|V(T_{x_a, \{y_1, \dots, y_r\}}| \leq 3/4|V(T)|$

From the above upper bound, we conclude that $r<q$.
Also, according to the algorithm, we have a lower bound
$|V(T_{x_a, \{y_1, \dots, y_r\}}| > 1/3(|(V(T)|-1)+1$ 
We conclude that $|V(T_{x_a, \{y_{r+1}, \dots, y_q\}}| <2/3(|V(T)-1)+1$
leading us to the same lower bound of $3/4|V(T)|$. 
Thus $x_a$, $\{y_1, \dots, y_r\}$ and $\{y_{r+1}, \dots, y_q\}$
satisfy the conditions of the lemma.
  
\end{proof}

\RESTATEbalancedvert*
\begin{proof}
Apply Lemma \ref{treesplit1} to $T[X]$
to obtain $(t,W_1,W_2)$ as per the statement of the lemma
that, in particular, can be efficiently computed.
We set $Y_1=W_1$ and $Y_2=N_T(t) \setminus Y_1$
and claim that $(t,Y_1,Y_2)$ is the triple claimed by the theorem. 
To prove that it is enough to demonstrate that
$X_{t,Y_i}=V(T[X]_{t,W_i})$ for each $i \in \{1,2\}$. 

This will be immediate from the following statement.
Let $Y \subseteq N_T(t)$ and let $W=Y \cap X$.
Then $X_{t,Y}=V(T[X]_{t,W})$. 
Indeed, let $t' \in X_{t,Y}$. This means that in
the path $P$ from $t$ to $t'$ in $T$, the second vertex belongs 
to $Y$. Due to the connectedness of $T[X]$, $V(P)$ wholly lies in 
$X$. In particular, this means that the second vertex of $P$ belongs
to $W$ and $P$ is a path of $T[X]$. Consequently, 
$t' \in V(T[X]_{t,W})$. Conversely, suppose $t' \in V(T[X]_{t,W})$.
This means that $T[X]$ (and hence $T$) has a path from $t$ to $t'$ whose second vertex
belongs to $W$ and hence to $Y$. It follows from $t' \in V(T_{t,Y})$. 
Since $t' \in X$ (due to being a vertex of $T[X]$), we conclude that
$t'\in X_{t,Y}$. 
\end{proof}

\begin{lemma} \label{rejectjust}
Let $H$ be a hypergraph, $A,B \subseteq V(H)$, $k_0>0$ an integer,
$(T,{\bf B})$ a \textsc{td} of $H$ and $X \subseteq V(T)$.
Suppose that $\sep(H,A,B,k_0,(T,{\bf B}),X) \neq \emptyset$.
Let $t \in V(T)$, $Y_1,Y_2$ be the partition of $N_T(t)$. 
Assume that $\sep(H,A,B,k_0,(T^+_{t,Y_i},{\bf B}^+_{t,Y_i}),X_{t,Y_i})=\emptyset$
for each $i \in \{1,2\}$. 
Then there are $W \subseteq {\bf B}(t)$, a weak partition $C_1,C_2$
into unions of connected components of $H[{\bf B}(t) \setminus W]$ and positive integers $k_1,k_2$,
$k_1+k_2 \leq k_0$ so that 
$\sep_i=\sep(H_{t,Y_i}^{-W}, (A \setminus W)_{t,Y_i} \cup C_1,(B \setminus W)_{t,Y_i} \cup C_2,k_i,
 (T_{t,Y_i},{\bf B}_{t,Y_i}^{-W}),X_{t,Y_i}) \neq \emptyset$ for each $i \in \{1,2\}$. 
\end{lemma}

\begin{proof}
Let $W_0 \in \sep(H,A,B,k_0,(T,{\bf B}),X)$ .
Let $W=W_0 \cap {\bf B}(t)$, $W_i=W_0 \cap {\bf B}(V(T_{t,Y_i})) \setminus W$ for each $i \in \{1,2\}$. 
The following statement is immediate.
\begin{claim} \label{split30}
$W_0$ is the disjoint union of $W_1,W_2,W$. 
\end{claim}
 
Since $W_0 \subseteq {\bf B}(X)$, it follows that
$W_i \cup W\subseteq {\bf B}(X_{t,Y_i})$ for each $i \in \{1,2\}$. 
Taking into account the relevant definitions we conclude that
\begin{equation} \label{splitcut31}
W_i \cup W \subseteq {\bf B}^+_{t,Y_i}(X_{t,Y_i}) \qquad \forall i \in \{1,2\} 
\end{equation}

and that

\begin{equation} \label{split32}
W_i \subseteq {\bf B}^{-W}_{t,Y_i}(X_{t,Y_i}) \qquad \forall i \in \{1,2\}
\end{equation} 

We next observe that both $W_1$ and $W_2$ are non-empty.
Indeed, suppose that say $W_2=\emptyset$.
Then it follows from the combination of Claim \ref{split30}
and \eqref{splitcut31} that 
$W_0 \in \sep(H,A,B,k_0,(T^+_{t,Y_1},{\bf B}^+_{t,Y_1}),X_{t,Y_1})$ 
in contradiction to our assumption. Let $k_i=\rho(W_i)$ for each $i \in \{1,2\}$.
We thus have just observed that 
\begin{equation} \label{split33}
k_i>0 \qquad \forall i \in \{1,2\}
\end{equation}

By the properties of tree decompositions
$\rho(W_1 \cup W_2)=\rho(W_1)+\rho(W_2)$.
As $W_1 \cup W_2 \subseteq W_0$, we conclude that

\begin{equation} \label{split34}
k_1+k_2 \leq k_0
\end{equation}

Let $C_1$ be the vertices of ${\bf B}(t) \setminus W$ reachable from $A$ in $H \setminus W_0$.
Let $C_2=({\bf B}(t) \setminus W) \setminus C_1$. 
Since $W=W_0 \cap {\bf B}(t)$, it is clear that $C_1,C_2$ is a weak partition of ${\bf B}(t) \setminus W$
into unions of connected components of $H[{\bf B}(t) \setminus W]$. 
It also follows from definitions of $C_1$ and $C_2$ that

\begin{equation} \label{split35}
W_0 \in \sep(H, A\cup C_1, B\cup C_2,(T,{\bf B}),X)
\end{equation}

We next observe that, for each $i \in \{1,2\}$,
$W_i$ separates $(A \setminus W)_{t,Y_i} \cup C_1$ 
from $(B \setminus W)_{t,Y_i} \cup C_2$ in $H_{t,Y_i} \setminus W$.
Indeed, suppose that there is a path $P$ in say $H_{t,Y_1} \setminus (W \cup W_1)$ 
between $u_1 \in (A \setminus W)_{t,Y_1} \cup C_1$ and $u_2 \in (B \setminus W)_{t,Y_2} \cup C_2$.
By definition, $H_{t,Y_1} \setminus W$ is disjoint with $W_2$.
Hence, $P$ is a path in $H \setminus W_0$ between $A \cup C_1$ and $B \cup C_2$
in contradiction to \eqref{split35}. 
Combining this observation with \eqref{split32} we observe that
\begin{equation} \label{36}
W_i \in \sep_i \qquad \forall i \in \{1,2\}
\end{equation}
The lemma follows from the combination 
of \eqref{split33}, \eqref{split34} and \eqref{split35}
\end{proof}

\subsection{The main proof.}
\begin{claim}  \label{clmdecomp1}
Theorem \ref{fptsep} holds for $X \leq 2$. 
\end{claim}

\begin{claimproof}
In this case the algorithm $\mathit{AppSep}$ runs 
$\mathit{SmallSep}$. 
It is immediate from Theorem \ref{setapprox}
that the algorithm runs \fpt time in $p$ and $d$ 
If a set is returned, it follows from the description
of $\mathit{SmallSep}$ that the returned set is an $(A,B)$-separator
and it follows from Theorem \ref{setapprox} that
the edge cover number is the set is at most
$(3k+d+1) \cdot k_0$. 

Assume that $\mathit{SmallSep}$ returns \textbf{Reject}.
One reason for that is that one of the runs
of $SetApp$ returns \textbf{Reject}.
In this case, the correctness of the output 
follows from Theorem \ref{setapprox}. 

Otherwise, according to the algorithm description
\textbf{Reject} is returned because no element of $Sets$
is an $(A,B)$-separator. By Theorem \ref{setapprox},
existence of a separator with edge cover number at most
$k_0$ would imply its superset to be included in $Sets$.
We conclude that such a separator does not exist and
the output is correct.
\end{claimproof}

\begin{claim} \label{clmdecomp2}
Suppose that $\mathit{AppSep}(H,A,B,k_0,(T,{\bf B}),X)$
applies itself recursively to a tuple 
$(H',A',B',k;,(T',{\bf B'}.X')$.
Then the tuple satisfies the input constraints as specified
in Algorithm \ref{alg:sepgen}. 
\end{claim}

\begin{claimproof}
We only need to establish that $(T',{\bf B'}$ is
a tree decomposition of $H'$ and that
all the elements of $V(T') \setminus X'$ are 
leaves of $T'$. 
This follows from Theorem \ref{validtw}. 
For the recursive calls in lines $4$ and $7$ of Algorithm \ref{alg:sepgen}, use the second
and the last statements.
For the recursive calls in lines $14$ and $15$ of Algorithm \ref{alg:sepgen},
use the first, the third, and the last statements. 
\end{claimproof}

In the rest of the proof we often use 
induction on $|X|$. Specifically, we need to assume
that, by the induction assumption, the considered property is correct for 
recursive calls made by the considered call of $\mathit{AppSep}$.
This use of induction assumption is valid because of Claim \ref{clmdecomp2}.
In the rest of the proof, we use induction assumption
without explicit reference to Claim \ref{clmdecomp2}.

\begin{claim} \label{clmdecomp3}
$\mathit{AppSep}$ runs \fpt time parameterized
by $p$ and $d$. 
\end{claim}

\begin{claimproof}
First, we verify fixed-parameter tractability of a single
recursive application of $\mathit{AppSep}$. We refer to the lines
of Algorithm \ref{alg:sepgen}.

Lines $1$ and $2$ are \fpt by Claim \ref{clmdecomp1}. Line 10 is \fpt
by Theorem \ref{setapprox}. In particular if the application of
$SetApp$ at line $10$ does not reject there is an \fpt upper bound on
the size of $Sets$. Hence, for the loop at line $13$, we only need to
demonstrate an \fpt upper bound on the number of partitions $C_1$ and
$C_2$. But this immediately follows from the number of connected
components of $H[{\bf B}(t) \setminus W]$ being at most
$\rho({\bf B}(t))$, which is at most $p$ by assumption.

We observe that the number of recursive applications depend 
on $p,d,X,k_0$. Since $p,d$ do not change during the recursive
call, we consider them fixed and provide as subscripts rather than
the arguments for the number of iterations. Thus we denote 
the function by $T_{p,d}(k_0,m)$ where $m=|X|$.
It is also easily established by induction that the function is \emph{monotone},
that is $k'_0 \leq k_0$ and $m' \leq m$, 
$T_{p,d}(k'_0,m') \leq T_{p,d}(k_0,m)$. 

By the monotonicity and the choice of $(t,Y_1,Y_2)$,
the number of recursive applications made by the call at
line $4$ is at most $T_{p,d}(k_0,3/4m)$ 
 and the same is true for line $7$. 
Due to the same reason the number of recursive applications
made by the call at line $14$ is at most $T_{p,d}(k_0-1,3/4m)$
and the same is true for line $15$. 
The total number of recursive applications made within
the loop of line $13$ is upper bounded by a function $g_{p,d}(k_0)$. 
Hence $T_{p,d}(k_0) \leq 2T_{p,d}(k_0,3/4)+g_{p,d}(k_0)T_{p,d}(k_0-1,3/4m)$. 
By Lemma \ref{numrecur}, 
$T_{p,d}(k_0,m) \leq (2.7^{k_0-1} \cdot \prod_{i=2}^{k_0} g_{p,d} (i)) m^3$.
Since $k_0 \leq k$, there is an \fpt upper bound for this function
parameterized by $p$ and $d$. 
\end{claimproof}

\begin{claim} \label{clmdecomp4}
Suppose that $\mathit{AppSep}(H,A,B,(T,{\bf B}), k_0, X)$ does not reject
and let $Out$ be the output.
Then $Out \in sep(H,A,B,(3k+d+1)(2k-1)k_0,(T,{\bf B}),X)$.  
\end{claim}

\begin{claimproof}
It can be verified by a direct inspection that
$Out \in {\bf B}(X)$. Thus we need to verify that
$Out$ is an $(A,B)$-separator of $H$ and that
$\rho(Out) \leq (3k+d=1)(2k-1)k_0$. 

Let us first prove that $Out$ is an $(A,B)$-separator.
If the output is returned on line $2$ then the claim is 
immediate from Claim \ref{clmdecomp1}. 
If $Out$ is returned on line $6$ or line $9$ then the statement
follows from the induction assumption applied to the respective 
runs at lines $4$ and $7$. 
Otherwise, $Out$ is returned at line $17$. 
Then the statement follows from the induction assumption applied to lines $14$
and $15$ and from Theorem \ref{composedsep}. 

In order to prove the required upper bound
on the edge cover of $Out$, we define the \emph{forming tree}
$F(H,A,B,k_0,T,{\bf B},X)$. The nodes of the tree are associated 
with the set of parameters of the corresponding recursive call
and, possibly, with a set. 

We define the tree recursively. 
Suppose first that $|X| \leq 2$.
Then $F(H,A,B,k_0,T,{\bf B},X)$ consists of a single
node associated with the tuple $(H,A,B,k_0,T,{\bf B},X)$
and with the set returned by $\mathit{SmallSep}$ on this tuple of parameters.

Otherwise, suppose that $Out$ is returned on line $6$.
Then the root of $F(H,A,B,k_0,T,{\bf B},X)$ is associated 
with the tuple $(H,A,B,k_0,T,{\bf B},X)$ and not associated 
with any set.
The root has a single child which is the root of 
$F(H',A,B,k_0,T',{\bf B'},X')$ where  $(H',A,B,k_0,T',{\bf B'},X')$
is the tuple of parameters for the recursive application of
$\mathit{AppSep}$ on line 4. 
If $Out$ is returned on line $9$ then $F(H,A,B,k_0,T,{\bf B},X)$
is defined analogously with line $7$ used instead of line $4$.

It remains to assume that $Out$ is returned on line $17$.
Fix $W$, $k_1,k_2$ considered at the corresponding iteration 
of the loop at line $13$. 
Let $(H',A',B',k_1,T',{\bf B'},X')$ and $(H'',A'',B'',k_2,T'',{\bf B''},X'')$
be the tuples of parameters on which $\mathit{AppSep}$ recursively runs
in the respective lines $14$ and $15$ of the considered iteration.
Then the root of $F(H,A,B,k_0,T,{\bf B},X)$ is associated with $(H,A,B,k_0,T,{\bf B},X)$
and with $W$. The root has two children. The subtree rooted by one of them is
$F(H',A',B',k_1,T',{\bf B'},X')$  and the subtree rooted by the other one
is $F(H'',A'',B'',k_2,T'',{\bf B''},X'')$. 

It is not hard to observe by induction that $Out$ is the union of sets associated
with the nodes of $F(H,A,B,k_0,T,{\bf B},X)$.
We observe that the edge cover of the union is at most
$(3k+d+1)(2k_0-1)k_0 \leq (3k+d+1)(2k-1)k_0$.
The proof is by induction on the number of nodes of 
$F(H,A,B,k_0,T,{\bf B},X)$. In case of a single node, this is immediate
from Claim \ref{clmdecomp1}.

If the root has one child then the root, by construction, is not associated
with any set. Therefore, the upper bound is the same as for the subtree
rooted by the child and hence follows by the induction assumption.

Finally, if the root has two children
then the set associated with the root is an element of the family
of sets produced by $SepApp$ whose
last parameter is at most $k_0$. It follows from Theorem \ref{setapprox} 
that the edge cover number of this set is at most $(3k+d+1)k_0$.
The last parameters of the tuples associated with the children
are $k_1,k_2>0$ such that $k_1+k_2 \leq k_0$.
By the induction assumption, the unions of sets 
associated with the subtrees rooted by the children are
of edge cover at most $(3k+d+1)(2k_i-1)k_i$ for $i \in \{1,2\}$
So the edge cover of the union of sets associated with 
$F(H,A,B,k_0,T,{\bf B},X)$ is at most
$(3k+1+d)((2k_1-1)k_1+(2k_2-1)k_2+k_0) \leq (3k+d+1)k_0(2k_1-1+2k_2-1+k_0) \leq
(3k+d+1)(2k_0-1)k_0$ as required.
\end{claimproof}

\begin{claim} \label{clmdecomp5}
If $AppSep(H,A,B,T,{\bf B},k_0,X)$ returns \textbf{Reject}
then either the set of relevant separators $sep(H,A,B,k_0,T,{\bf B},X)$ is empty,
or $ghw(H)>k$,
\end{claim}

\begin{claimproof}
Let us say that a \textbf{Reject} returned by $\mathit{SmallSep}$ is \emph{triggered}
by $SetApp$ if the \textbf{Reject} is returned at line $8$ of  
Algorithm \ref{alg:smallsep}. 

Let us say that a \textbf{Reject} returned by $\mathit{AppSep}$ is \emph{triggered} by $SetApp$
if one of the following three conditions hold. 
\begin{enumerate}
\item The condition of the first line is true and the
corresponding run of $\mathit{SmallSep}$ returns \textbf{Reject} triggered by a large \ghw.
\item The \textbf{Reject} is returned at line $9$
\item The \textbf{Reject} is returned at line $18$ and at least one recursive
application made by $\mathit{AppSep}$ before line $18$ is triggered by $SetApp$. 
\end{enumerate}

We observe that if a \textbf{Reject} return by $\mathit{AppSep}$ is triggered by $SetApp$
then $ghw(H)>k$. 
This is immediate from the description and Theorem \ref{setapprox} in case
the first or the second condition of the above list are true.
In case of the third condition, it is not hard to argue by induction that 
the \ghw of an induced subgraph of $H$ is greater than $k$ and so
this holds for $H$ itself. 

We proceed to observe that if the \textbf{Reject} return by $\mathit{AppSep}$
is not triggered by $SetApp$ then $\sep(H,A,B,k_0,T,{\bf B},X)=\emptyset$
The proof is by induction on $|X|$. 
If $|X| \leq 2$,the statement is immediate from Claim \ref{clmdecomp1}.
Otherwise, the \textbf{Reject} is returned at line $18$ of Algorithm \ref{alg:sepgen}.
In order to trigger the return at line $18$, \textbf{Reject} must be returned at lines
$6$ and $9$. By the induction assumption applied to the recursive calls 
at lines $4$ and $7$, this means that
$\sep(H,A,B,k_0,(T^+_{t,Y_i},{\bf B}^+_{t,Y_i}),X_{t,Y_i})=\emptyset$ for each $i \in \{1,2\}$.
If we assume that $\sep(H,A,B,k_0,(T,{\bf B}),X) \neq \emptyset$ then the premises
of Lemma \ref{rejectjust} are satisfied. Combining the lemma with the induction
assumption, we observe that a non-rejection output must be returned at
one of the iterations of the loop at line 13 contradicting our assumption that
the algorithm reaches line 18.  
\end{claimproof}

Theorem \ref{fptsep} is immediate from Claims \ref{clmdecomp3}, \ref{clmdecomp4},
and \ref{clmdecomp5}.

\end{document}